\documentclass[11pt]{article}
\usepackage{amsmath,amssymb,amsthm,verbatim,bm,geometry,fullpage}
\allowdisplaybreaks
\usepackage[colorlinks]{hyperref}

\newtheorem{theorem}{Theorem}
\newtheorem{lemma}{Lemma}

\newtheorem{proposition}{Proposition}

\newtheorem{definition}{Definition}

\newtheorem{fact}{Fact}

\newcommand{\ket}[1]{| #1 \rangle}
\newcommand{\bra}[1]{\langle #1 |}

\newcommand{\C}{\mathbb{C}}
\newcommand{\E}{\mathbb{E}}
\newcommand{\I}{\mathbb{I}}
\newcommand{\N}{\mathbb{N}}
\newcommand{\R}{\mathbb{R}}
\newcommand{\sphere}{\mathbb{S}}
\newcommand{\U}{\mathbb{U}}
\newcommand{\Z}{\mathbb{Z}}

\newcommand{\cD}{\mathcal{D}}
\newcommand{\cM}{\mathcal{M}}
\newcommand{\cN}{\mathcal{N}}
\newcommand{\cS}{\mathcal{S}}
\newcommand{\cU}{\mathcal{U}}
\newcommand{\cW}{\mathcal{W}}

\newcommand{\Tr}{\mathrm{Tr}\,}
\newcommand{\op}{\mathrm{op}}
\newcommand{\tuple}{\mathrm{vec}}
\newcommand{\Haar}{\mathrm{Haar}}
\newcommand{\erf}{\mathrm{erf}}
\newcommand{\Sing}{\mathrm{Sing}}

\newcommand{\one}{\leavevmode\hbox{\small1\kern-3.8pt\normalsize1}}

\title{{\bf
Approximate unitary $n^{2/3}$-designs give rise to quantum channels with
super additive classical Holevo capacity
}}

\author{
Aditya Nema${}^*$
\and
 Pranab Sen\thanks{
School of Technology and System Science, 
Tata Institute of Fundamental Research, Mumbai, India.
Email: {\sf
\{aditya.nema30, pranab.sen.73\}@gmail.com
}
}
}

\date{}

\begin{document}
\maketitle

\begin{abstract}
In a breakthrough, Hastings~\cite{hastings_2009} showed that there
exist quantum channels whose classical Holevo capacity is 
superadditive i.e.
more classical information can be transmitted by quantum
encoding strategies entangled across multiple channel uses as compared to
unentangled quantum encoding strategies. Hastings' proof used Haar
random unitaries to exhibit superadditivity. In this paper
we show that a unitary chosen uniformly at random from an approximate 
$n^{2/3}$-design gives rise to a quantum channel with superadditive
classical Holevo capacity, where $n$ is the dimension of the unitary 
exhibiting the Stinespring dilation of the channel superoperator.

We follow the geometric functional analytic approach 
of Aubrun, Szarek and 
Werner~\cite{aubrun_szarek_werner_2010_main}
in order to prove our result. More precisely we prove a sharp 
Dvoretzky-like theorem stating that, with high probability under the
choice of a unitary from an approximate $t$-design, random 
subspaces of large dimension make a Lipschitz function take almost 
constant value. 
Such theorems were known earlier only for Haar random unitaries.
We obtain our result by appealing to Low's 
technique~\cite{low_2009} for proving concentration of measure for
an approximate $t$-design, combined with a stratified analysis of
the variational behaviour of Lipschitz functions on the unit sphere in 
high 
dimension. The stratified analysis is the main echnical advance
of this work.

Haar random unitaries require at least $\Omega(n^2)$ random bits
in order to describe them with good precision. In contrast, there exist
exact $n^{2/3}$-designs using only $O(n^{2/3} \log n)$ random
bits \cite{Kuperberg}. Thus, 
our work can be viewed as 
a partial derandomisation of Hastings' result, and a step towards the
quest of finding an explicit quantum channel with superadditive classical
Holevo capacity.

Finally we also show that for any $p > 1$, approximate unitary 
$(n^{1.7} \log n)$-designs
give rise to channels violating subadditivity of R\'{e}nyi
$p$-entropy. In addition to stratified analysis, the proof of this
result uses a new technique of approximating a
monotonic differentiable function defined on a closed bounded interval
and its derivative by moderate 
degree polynomials which
should be of independent interest.
\end{abstract}

\section{Introduction} 
For the past two decades, additivity conjectures have been extensively 
studied in quantum information theory e.g. 
\cite{Bennett_EtAl, Pomeransky, AmosovEtAl, Osawa, Shor, 
hayden_winter_2008}.
In this paper, we concentrate on the issue of additivity of classical
Holevo capacity of a quantum channel $\Phi$, denoted henceforth by 
$C(\Phi)$. 
The quantity $C(\Phi)$
is the number of classical bits of information per channel use that
can reliably be transmitted in the limit of infinitely many independent
uses of $\Phi$. Capacities of classical memoryless 
channels are known to be additive, that is, the capacity of two
channels $\Phi$ and $\Psi$, used independently, is the sum of the 
individual capacities. In other words,
$C(\Phi \otimes \Psi) = C(\Phi) + C(\Psi)$.
This additivity property leads to 
a single letter characterization of the capacity of classical channels
viz. the capacity is nothing but the mutual information between the
input and channel output maximised over all possible input distributions
for one channel use
\cite{vajda_shannon_weaver_1950}.
For a long time, in analogy with the classical setting, it was generally 
believed that the classical Holevo capacity of a quantum channel is
additive. In fact, this belief was proven to be true for several classes
of quantum channels e.g.
\cite{KingUnital, FujiwaraEtAl, KingDepolarising, ShorEntang, 
KingEtAl}. Thus, it came as a major surprise to 
the community
when Hastings, in a major breakthrough, showed that there are indeed
quantum channels with superadditive classical Holevo capacity 
\cite{hastings_2009} i.e. there are quantum channels 
$\Phi$, $\Psi$ such that $C(\Phi \otimes \Psi) > C(\Phi) + C(\Psi)$.

Hastings' proof proceeds by showing that a
Haar random unitary leads to such channels with high probability,
in the sense that the unitary, when viewed suitably, is the 
Stinespring dilation of a quantum channel with superadditive classical
Holevo capacity.
The drawback of using Haar random unitaries is that they are inefficient
to implement. In fact, it takes at least 
$\Omega(n^2 \log (1/\epsilon))$
random bits in order to pick an $n \times n$ Haar random unitary to 
within a precision of $\epsilon$ in the $\ell_2$-distance 
\cite{Vershynin}. Hence, it is of considerable interest to find an explicit
efficiently implementable unitary that gives rise to a quantum channel
with superadditive classical Holevo capacity. 

In this paper, we take the first 
step in this direction. We show that with high probability 
a uniformly random $n \times n$ unitary from an 
approximate $n^{2/3}$-design
leads to a quantum channel with 
superadditive classical Holevo capacity. Though no efficient algorithms for
implementing  approximate $n^{2/3}$-designs are known, nevertheless,
it is known that a uniformly random unitary from an exact
$n^{2/3}$-design can be sampled using
only $O(n^{2/3} \log n)$ random bits \cite[Theorem 3.3]{Kuperberg}.
Also, efficient constructions of approximate $(\log n)^{O(1)}$-designs
are known \cite{sen:zigzag, brandao2012local}.
Thus, our work can be viewed as 
a partial derandomisation of Hastings' result, and a step towards the
quest of finding an explicit quantum channel with superadditive classical
Holevo capacity.

Hastings' proof was considerably simplified by Aubrun, Szarek and
Werner~\cite{aubrun_szarek_werner_2010_main} who showed that existence
of channels with subadditive minimum output von Neumann entropy follows
from a sharp Dvoretzky-like theorem which states that, under the
Haar measure, random subspaces of large dimension make a Lipschitz function
take almost constant value. 
Dvoretzky's original theorem \cite{Dvoretzky} 
stated that any centrally
symmetric convex body can be embedded with low distortion into a
section of a high dimensional unit $\ell_2$-sphere. 
Milman \cite{Milman} extended Dvoretzky's theorem by proving that,
with high probability, Haar random subspaces of an appropriate dimension
make a Lipschitz function take almost constant value. Dvoretzky's theorem
becomes the special case of Milman's theorem where the Lipschitz 
function happens to be norm induced by the centrally symmetric convex
body i.e. the norm under which the convex body becomes the unit ball.
Milman's work started a whole body of research sharpening the various
parameters of the extended Dvoretzky theorem 
e.g. \cite{Schechtman, Gordon} etc. However, all these works
use Haar random subspaces. A Haar random subspace of $\C^n$ of
dimension $d$
can be obtained by applying a Haar random unitary to a fixed subspace
of dimension $d$ e.g. the subspace spanned by the first $d$ standard
basis vectors of $\C^n$. Our work is the first one to replace the
Haar random unitary in any Dvoretzky-type theorem by a uniformly 
random unitary chosen from an
approximate $t$-design for a suitable value of $t$. In other words,
our main technical result is an Aubrun-Szarek-Werner style result for
approximate $t$-designs instead of Haar random unitaries.
As a corollary, we obtain the subadditivity of minimum output von
Neumann entropy for unitaries chosen from an approximate $n^{2/3}$-design.
As another corollary, we obtain the subadditivity
of minimum output R\'{e}nyi $p$-entropy for all $p > 1$ for 
quantum channels arising from unitaries
chosen from an approximate unitary 
$(n^{1.7} \log n)$-design. 
Such a unitary can in fact be chosen from an exact 
$(n^{1.7} \log n)$-design using only 
$n^{1.7} (\log n)^2$ random
bits \cite{Kuperberg}, which is
much less than $\Omega(n^2)$ random bits required to choose a
Haar random unitary. Subadditivity
of minimum output R\'{e}nyi $p$-entropy for all $p > 1$ 
was originally proved for Haar random unitaries
by Hayden and Winter \cite{hayden_winter_2008}.

To prove our main technical result, we use a concentration of 
measure result by 
Low~\cite{low_2009} for approximate unitary $t$-designs,
combined with a stratified analysis of the variational
behaviour of Lipschitz functions on the unit sphere in high 
dimension. 
We need such a fine grained stratified analysis for the
following reason. Aubrun, Szarek and 
Werner~\cite{aubrun_szarek_werner_2010_main} worked with the function
$f(M) := \lVert MM^\dag - (I/k) \rVert_2$, where the argument $M$
is a $k^3$-tuple rearranged to form a $k \times k^2$ matrix. They
found subspaces of dimension $k^2$ where $f$ took almost constant value.
For this, they had to do a two step analysis. The global Lipschitz
constant of $f$ was $2$ which, under naive Dvoretzky type arguments,
would only guarantee the existence of subpaces of dimension 
$\frac{k^2}{\log k}$ where $f$ is almost constant. This does not
suffice to find a counter example to minimum output von Neumann
entropy. In order to shave off the $\log k$ term in the denominator,
they had to use several sophisticated arguments. One of them was the
observation that there is a high probablity subset $T$ of 
$\sphere_{\C^{k^3}}$
on which the Lipschitz constant of $f$ was $k^{-1/2}$. They exploited
this by their two step analysis, where they separately analysed the
behaviour of $f$ on $T$ and on $T^c$, and managed to shave off 
the $\log k$ term. For us, since we are working with designs, we
need the function to be a polynomial. Hence, instead of $f$, we have
to work with $f^2$. This seemingly trivial change introduces severe
technical difficulties. The main reason behind them is that the Lipschitz
constant of $f^2$ is about twice the Lipschitz constant for $f$
but the variation that we are looking to bound is around square of the
earlier variation! This contradiction lies at the heart of the technical
difficulty. In order to overcome this, we have to partition 
$\sphere_{\C^{k^3}}$ into a number of sets 
$\Omega_1, \Omega_2, \ldots, \Omega_{\log k}$,
called `layers',
with local Lipschitz constants for $f^2$ running as 
$k^{-3/2}, 2^3 k^{-3/2}, 3^3 k^{-3/2}, \ldots, (\log k)^3 k^{-3/2}$. 
We have to bound
the variation of $f^2$ individually on $\Omega_i$ as well as put
them together to bound the variation on large subspheres of 
$\sphere_{\C^{k^3}}$. This leads to a challenging stratified 
analysis, which 
forms the main technical advance of this paper. 

Another tool developed in this work which should find use in other
situations also, is a systematic way to approximate a monotonic
differentiable function and its derivative using moderate degree
polynomials. This tool is crucially used to prove strict subadditivity
of R\'{e}nyi $p$-entropy for any $p > 1$ for channels whose unitary
Stinespring dilation is chosen from an approximate design instead of
a Haar random unitary.

The power of our 
stratified analysis shows up in the consequence that the dimension of 
the subspace on which the Lipschitz function is almost constant depends
only on the smallest local Lipschitz constant, provided some mild
niceness conditions are satisfied. This gives larger dimensional subspaces
than a naive analysis which would depend on the global Lipschitz 
constant. In fact, the stratified analysis
allows us to prove a sharper Dvoretzky-type theorem even for the 
Haar measure. As a result, we can recover Aubrun, Szarek and Werner's
result for the function $f$ directly and elegantly
instead of applying their Dvoretzky-type result twice which is rather
messy. Another powerful consequence of our stratified analysis is that 
with probability
exponentially close to one random, over Haar measure or $t$-design
measure, large subspaces 
make the Lipschitz function almost constant. In contrast,
Aubrun, Szarek and Werner could only guarantee constant probability close
to one for the Haar measure, and they did not consider $t$-designs. 
They  also stated without providing details that the existence probability
could be made exponentially close to one using a deep Levy-type lemma for
unitary matrices. In contrast our stratified analysis uses only 
the elementary Levy lemma for the unit sphere, yet it manages to prove
existence with probability exponentially close to one.

The rest of the paper is organised as follows. Section~\ref{sec:prelim}
contains notations, symbols 
definitions and preliminary tools required for the 
paper. 
Section~\ref{sec:main} states and proves the main technical theorems viz.
the stratified analyses for Haar measure and approximate $t$-designs.
Section \ref{sec:vonNeumannentropy} describes the 
application to subadditivity of minimum output von Neumann entropy.
Section~\ref{sec:Renyipentropy} describes the 
application to subadditivity of minimum output R\'{e}nyi $p$-entropy 
for $p > 1$. Section~\ref{sec:conclusion} concludes the paper and states
some open problems for future work.

\section{Preliminaries}
\label{sec:prelim}
All Hilbert spaces used in this paper are finite dimensional.
The $n$ dimensional space over complex numbers, $\C^n$, is endowed
with the standard inner product aka the dot product:
$\langle x, y \rangle := \sum_{i=1}^n x_i^* y_i$.
The unit radius sphere in $\C^n$ is denoted by $\sphere_{\C^n}$.
The symbol $\cM_{k,d}$ 
denotes the Hilbert space of $k \times d$ linear operators over the 
complex field under the Hilbert-Schmidt inner product
$\langle M, N \rangle := \Tr [M^\dag N]$, 
and $\cM_d := \cM_{d,d}$. 
Let $\cU(n)$ denote the set of $n \times n$ unitary matrices with
complex entries.
For a composite Hilbert space
$\C^k \otimes \C^d$, 
the notation $\Tr_{\C^d}[\cdot]$ denotes the operation of 
taking partial trace i.e.
tracing out the mentioned subsystem $\C^d$. We use $\Tr[\cdot]$ 
to denote  the trace of the underlying operator.
Fix standard bases for Hilbert spaces $A \cong \C^k$, $B \cong \C^d$. 
Let $\ket{e_i}^A$, $\ket{e_i}^B$ denote standard basis vectors of
$A$, $B$ respectively. Any vector $x \in A \otimes B$ can be
written as $x = \sum_{i,j} \alpha_{ij} \ket{e_i}^A \otimes \ket{e_j}^B$.
We use $\op_{d \rightarrow k}(x)$ to denote the operator
$\sum_{i,j} \alpha_{ij} \ket{e_i}^A \otimes \bra{e_j}^B$
in $\cM_{k, d}$. Conversely, given an operator 
$M = \sum_{ij} m_{ij} \ket{e_i}^A \otimes \bra{e_j}^B$ in 
$\cM_{k,d}$, we let 
$\tuple(M) := \sum_{ij} m_{ij} \ket{e_i}^A \otimes \ket{e_j}^B$ denote
the vector in $\C^k \otimes \C^d$.

For Hermitian positive semidefinite operators $M$, we define
$M^\alpha$ for any $\alpha > 0$ to be the unique Hermitian
operator obtained by keeping the eigenbasis same and
taking the $\alpha$th power of the eigenvalues. We can define
$\log M$ similarly.
For $p > 1$, the notation
$\lVert M \rVert_p$ denotes the Schatten $p$-norm 
of the matrix $M$, which is
nothing but the $\ell_p$-norm of the vector of its singular values. 
Alternatively, $\lVert M \rVert_p = (\Tr [(M^\dagger M )^{p/2}])^{1/p}$.
Then $p=2$ gives the Hilbert Schmidt 
norm aka the Frobenius norm which is nothing but 
$\lVert M \rVert_2 = \lVert \tuple(M) \rVert_2$.
Also, $p = \infty$ gives the operator norm aka spectral norm
which is nothing but 
$
\lVert M \rVert_\infty = 
\max_{v: \lVert v \rVert_2 = 1} \lVert M v \rVert_2.
$

Unless stated otherwise, the symbol $\rho$ denotes a quantum 
state aka density matrix which is nothing but a Hermitian, positive
semidefinite matrix with unit trace. A rank one density matrix is called
a pure state. By the spectral theorem, any density matrix is a convex
combination of pure states.
The notation $\cD(\C^d)$ denotes the convex 
set of all $d \times d$ density matrices. We use $\ket{\cdot}$ to denote 
a unit vector. By a slight abuse of notation, we shall often use a unit
vector $\ket{\psi}$ to denote a pure state $\ket{\psi}\bra{\psi}$.
A linear mapping $\Phi: \cM_m \to \cM_d $ is called 
a superoperator. A superoperator is trace preserving if 
$\Tr \Phi(M) = \Tr M$ for all $M \in \cM_m$. It is said to be 
positive if $\Phi(M)$ is positive semidefinite for all positive 
semidefinite $M$. Furthermore, $\Phi$ is said to be completely positive
if $\Phi \otimes \I$ is a positive superoperator for identity
superoperators $\I$ of all dimensions.
Completely positive and trace preserving (CPTP) superoperators
are referred to as quantum channels.
Unless stated otherwise, $\Phi$, $\Psi$ are used to denote 
quantum channels. 

A compact convex set $\cS$ in $\C^n$ is called a convex body. 
The radius $r(\cS)$ of a convex body $\cS$ is defined as
\[
r(\cS) := \min_{x \in \cS} \max_{y \in \cS} \lVert x - y \rVert_2.
\]
Any point $x \in \cS$ achieving the minimum above is said to be a
centre of $\cS$.
The convex body $\cS$ is said
to be centrally symmetric iff for every $x \in \C^n$,
$x \in \cS \leftrightarrow -x \in \cS$. The zero vector is a centre
of a centrally symmetric convex body. A centrally symmetric convex body
lying in $\C^n$
can be thought of as the unit sphere of a suitable notion of norm
in $\C^n$. Conversely for any norm in $\C^n$, the unit sphere under
the norm forms a centrally symmetric convex body. 

\subsection{Entropies and norms}
\begin{definition}
The von Neumann entropy of a  quantum state 
$\rho$ is defined as 
\[
S(\rho) := -\Tr [\rho \log \rho].
\]
For all $p > 1$, the R\'{e}nyi $p$-entropy of a quantum state 
$\rho$ is defined as 
\[
S_p(\rho) := \frac{1}{1-p} \log \Tr \rho^p =
-\frac{p}{p-1} \log \lVert \rho \rVert_p.
\]
\end{definition}
It turns out that
$
S(\rho) = \lim_{p \downarrow 1} S_p(\rho) =: S_1(\rho).
$
\noindent
Also, it can be shown that for $p \geq 1$, $S_p(\cdot)$ 
is concave in its argument.
\begin{definition}
For $p \geq 1$, the minimum output R\'{e}nyi $p$-entropy of 
a quantum channel $\Phi$ is defined as :
\[
S_p^{\mathrm{min}}(\Phi) :=
\min_{\rho \in \cD(\C^m)} S_p(\Phi(\rho))
\]
\end{definition}
By an easy concavity argument it can be seen that above 
minimum is achieved on a pure state.
Equivalently, to obtain $S_p^{\mathrm{min}}(\Phi)$ for $p > 1$ 
we must maximise 
$\lVert \Phi(\rho) \rVert_p$ for all input states $\rho$. This quantity is
also known as the 
$1 \rightarrow p$ superoperator norm of superoperator 
$\Phi: \cM_m \rightarrow \cM_d$:
\[
\lVert \Phi \rVert_{1 \rightarrow p} 
:= \max_{M \in \cM_m: \lVert M \rVert_1 = 1} 
\lVert \Phi(M) \rVert_p.
\]
By an easy convexity argument it can be seen that the above
maximum is achieved on a pure state i.e.
\[
\lVert \Phi \rVert_{1 \rightarrow p} = 
\max_{x \in \C^m: \lVert x \rVert_2 = 1} \lVert \ket{x}\bra{x} \rVert_{p}.
\]
Thus, the additivity conjecture for minimal output p-R\'{e}nyi $p$-entropy,
$p > 1$, 
for quantum channels $\Phi$ and $\Psi$ is equivalent to 
multiplicativity of $1 \rightarrow p$-norms of quantum channels viz.
$
\lVert \Phi \otimes \Psi \rVert_{1 \rightarrow p} 
\stackrel{?}{=} \lVert \Phi \rVert_{1 \rightarrow p} \cdot
\lVert \Psi \rVert_{1 \rightarrow p}.
$
This equivalence will be used in Section~\ref{sec:Renyipentropy} to 
give a counter example to additivity 
conjecture for all $p > 1$ where the Stinespring dilation of the
quantum channel will be described 
from a unitary chosen uniformly at random from an approximate $t$-design.
The equivalent result for Haar random unitaries was originally 
proved by Hayden and Winter \cite{hayden_winter_2008}.

We heavily use the one-one correspondence between quantum channels and 
subspaces of 
composite Hilbert spaces, originally proved by  Aubrun, Szarek and Werner 
\cite{aubrun_szarek_werner_2010}, in this paper. 
Let $\cW$ be a subspace of 
$\C^k \otimes \C^d$ of dimension $m$. Identify $\cW$ with
$\C^m$ through an isometry 
$V : \C^m \to \C^k \otimes \C^d$ whose range is 
$\cW$. Then, the corresponding quantum channel 
$\Phi_\cW : \cM_m \to \cM_k$ is defined by
$
\Phi_\cW(\rho) := \Tr_{\C^d}(V \rho V^\dagger).
$
Using this equivalence and the fact that for $p > 1$ the 
$1 \rightarrow p$-superoperator norm is achieved on pure input
states, we can write \cite{aubrun_szarek_werner_2010}
\begin{equation}
\label{eq:redefpnorm}
\lVert \Phi_\cW \rVert_{1 \rightarrow p} = 
\max_{x \in \cW: \lVert x \rVert_2 = 1} 
\lVert \Tr_{\C^d} \ket{x}\bra{x} \rVert_p =
\max_{x \in \cW: \lVert x \rVert_2 = 1} 
\lVert \op_{d \rightarrow k}(x) \rVert_{2p}^2.
\end{equation}

In an important paper, Shor \cite{Shor} proved that several 
additivity conjectures
for quantum channels were in fact equivalent to the additivity
of minimum output von Neumann entropy of a quantum channel.
More specifically, Shor showed that if there is a quantum channel $\Phi$
whose minimum output von Neumann entropy is subadditive, then there
are quantum channels $\Psi_1$, $\Psi_2$ exhibiting superadditive
classical Holevo capacity viz.
$C(\Psi_1 \otimes \Psi_2) > C(\Psi_1) + C(\Psi_2)$. This equivalence
was used as a starting point by Hastings \cite{hastings_2009} in his 
proof that there are
channels with superadditive classical Holevo capacity. Aubrun, Szarek and
Werner \cite{aubrun_szarek_werner_2010_main}, as well as this paper
also have the same starting point. For this, we need the following
fact.
\begin{fact}[\mbox{\cite[Lemma 2]{aubrun_szarek_werner_2010_main}}] 
\label{fact:minoutputentropy}
Let a quantum
channel $\Phi_\cW: \cM_m \rightarrow \cM_k$ be described by a subspace
$\cW \leq \C^k \otimes \C^d$ of dimension $m$. Then,
\begin{eqnarray*}
S_{\mathrm{min}}(\Phi_\cW) 
& = &
\log k -
k \cdot 
\max_{\rho \in \cD(\C^m)} \lVert \Phi(\rho) - \frac{\one}{k} 
\rVert_2^2 \\
& = &
\log k -
k \cdot 
\max_{x \in \cW: \lVert x \rVert_2 = 1} 
\lVert 
(\op_{d \rightarrow k}(x)) (\op_{d \rightarrow k}(x))^\dagger 
- \frac{\one}{k} 
\rVert_2^2.
\end{eqnarray*}
\end{fact}

We will need the following result proved by Hayden and 
Winter \cite{hayden_winter_2008} 
that upper bounds $S_p^{\mathrm{min}}(\Phi \otimes \bar{\Phi})$ where
$\bar{\Phi}$ denotes the CPTP superoperator obtained by taking complex 
conjugate of the CPTP superoperator $\Phi$. 
\begin{fact}
\label{fact:maxeigenvalue}
Let $V : \C^m \rightarrow \C^k \otimes \C^d$ be an 
isometry describing the quantum channel 
$\Phi: \rho \mapsto \Tr_{\C^d} [V \rho V^\dagger]$. Let 
$\ket{\phi}$ denote the maximally entangled state in 
$\C^m \otimes \C^m$. Suppose $m \leq d$.
Then $(\Phi \otimes \bar{\Phi})(\ket{\phi} \bra{\phi})$ has a 
singular value 
not less than $\frac{m}{kd}$. 
Hence for all $p > 1$,
\[
\lVert \Phi \otimes \bar{\Phi} \rVert_{1 \rightarrow p} 
\geq 
\lVert \Phi \otimes \bar{\Phi} \rVert_{1 \rightarrow \infty} \geq 
\frac{m}{kd}. 
\]
Moreover,
\[
S_{\mathrm{min}}(\Phi \otimes \bar{\Phi}) \leq 
2 \log k - \frac{m}{kd} \log k + 
O \left(\frac{m}{kd} \log \frac{d}{m} + \frac{1}{k}\right).
\]
\end{fact}

\subsection{Polynomial approximation of monotonic functions}
We will need the following facts about step functions and their analytic
and polynomial approximations when we prove our result on strict
subadditivity of minimum output R\'{e}nyi $p$-entropy for 
channels chosen from approximate $t$-designs.
\begin{definition}
\label{def:stepfunction}
The {\em (Heaviside) step function} is a function 
$\R \rightarrow [0, 1]$ defined as follows:
\[
s(x) :=
\begin{array}{l l}
0 & \mbox{for $x < 0$} \\
\frac{1}{2} & \mbox{for $x = 0$} \\
1 & \mbox{for $x > 0$}.
\end{array}
\]
\end{definition}
\begin{definition}
\label{def:errorfunction}
The {\em error function} is a function $\R \rightarrow (-1, 1)$ defined
as follows:
\[
\erf(x) :=
\frac{2}{\sqrt{\pi}} \int_0^x e^{-t^2} \, dt.
\]
\end{definition}

\noindent
The error function is a monotonically increasing function. 
For positive $x$, $\erf(x)$ is nothing but the probability that the
normal distribution with mean $0$ and variance $1/2$ gives a point in
the interval $[-x, x]$. From the error function,
we get the so-called {\em sigmoid function}
$\Phi(x) := \frac{1}{2} + \frac{1}{2} \erf(x)$ which is nothing but
the cumulative distribution function of the above normal distribution.
The sigmoid function is a monotonically increasing function approximating
the step function in the following sense. Let $0 < \epsilon < 1$.
\begin{equation}
\label{eq:PhiVsS}
\Phi(x) ~
\begin{array}{l l l}
= & s(x) = \frac{1}{2} & \mbox{for $x=0$}, \\
> & s(x) = 0 & \mbox{for $x<0$}, \\
< & \frac{1}{2}  & \mbox{for $x<0$}, \\
< & s(x) = 1 & \mbox{for $x>0$}, \\
> & \frac{1}{2} & \mbox{for $x>0$}, \\
> & s(x) - \epsilon = 1 - \epsilon & 
\mbox{for 
$x>\sqrt{\ln \epsilon^{-1}}$
}, \\
< & s(x) + \epsilon = \epsilon & 
\mbox{for 
$x<-\sqrt{\ln \epsilon^{-1}}$
}, 
\end{array}
~~~
\Phi'(x) ~
\begin{array}{l l l}
= & \frac{1}{\sqrt{\pi}} & \mbox{for $x=0$}, \\
< & \frac{1}{\sqrt{\pi}} & \mbox{for $x \neq 0$}, \\
> & 0 & \mbox{for all $x$}, \\
< & \epsilon & 
\mbox{for 
$|x| >  \sqrt{\ln \epsilon^{-1}}$
}. 
\end{array}
\end{equation}
The last two statements for $\Phi(x)$ above hold for 
small $\epsilon$ and follow from 
the bound
$
1 - \Phi(x) \leq
\frac{1}{2 x \sqrt{\pi}} e^{-x^2}.
$

The error function has the following rapidly converging Maclaurin
series:
\[
\erf(x) =
\frac{2}{\sqrt{\pi}}
\sum_{i=0}^\infty (-1)^i \frac{x^{2i+1}}{i! (2i+1)}.
\]
It is obtained by integrating termwise the Maclaurin series
$
e^{-x^2} =
\sum_{i=0}^\infty (-1)^i \frac{x^{2i}}{i!}.
$
Since both the above series are  alternating series of positive and 
negative terms,
truncating the Maclaurin expansion of $\Phi(x)$ at $i = n$ for 
odd $n > x^2$ gives us 
a polynomial $p_n(x)$ of degree $2n+1$ such that
\begin{equation}
\label{eq:PhiVsP}
p_n(x) ~
\begin{array}{l l l}
= & \Phi(x) = \frac{1}{2} & \mbox{for $x=0$}, \\
> & \Phi(x) & \mbox{for $-\sqrt{n} \leq x < 0$}, \\
< & \Phi(x) & \mbox{for $0 < x \leq \sqrt{n}$}, \\
> & \Phi(x) - \epsilon 
& \mbox{for 
$
0 \leq x \leq \frac{\epsilon^{\frac{1}{2n}} \sqrt{n}}{2}
$
},\\
< & \Phi(x) + \epsilon 
& \mbox{for 
$
-\frac{\epsilon^{\frac{1}{2n}} \sqrt{n}}{2} \leq x \leq 0
$
}.
\end{array}
\end{equation}
Moreover, the derivative $p'_n(x)$ is a polynomial of degree $2n$ 
satisfying
\begin{equation}
\label{eq:PhiprimeVsPprime}
p'_n(x) ~
\begin{array}{l l l}
= & \Phi'(x) = \frac{1}{\sqrt{\pi}} & \mbox{for $x=0$}, \\
\leq & \Phi'(x) 
& \mbox{for 
$
-\sqrt{n} \leq x \leq \sqrt{n}
$
}, \\
> & \Phi'(x) - \epsilon 
& \mbox{for 
$
-\frac{\epsilon^{\frac{1}{2n}} \sqrt{n}}{2}
\leq x \leq \frac{\epsilon^{\frac{1}{2n}} \sqrt{n}}{2}
$
}.
\end{array}
\end{equation}
For the last two claims in Equation~\ref{eq:PhiVsP} and the last claim in 
Equation~\ref{eq:PhiprimeVsPprime}, 
we used Stirling's approximation
$
n^n e^{-n} < n!
$
which holds for all positive integers $n$.

We will also need to upper bound the sum of absolute values   of the
coefficients of $p_n(x)$, denoted by $\alpha(p_n(x))$. For this we
observe that
$
\alpha(p_n(x)) = 
|p_n(\sqrt{-1})| \leq
\frac{1}{2} +  \frac{e}{\sqrt{\pi}}.
$
We can now conclude that for $m > 0$, $0 \leq q \leq A$, 
\begin{equation}
\label{eq:alpha}
\begin{array}{rcl}
\alpha(p_n(m(x-q))) 
&   =  &
\frac{1}{2} +  \frac{1}{\sqrt{\pi}}
\alpha(
\sum_{i=0}^n (-1)^i
\frac{(m(x-q))^{2i+1}}{i! (2i + 1)}
) 
\;\leq\;
\frac{1}{2} +  \frac{1}{\sqrt{\pi}}
\alpha(
\sum_{i=0}^n
\frac{(m(x+q))^{2i+1}}{i! (2i + 1)}
) \\
&   =  &
\frac{1}{2} +  \frac{1}{\sqrt{\pi}}
\sum_{i=0}^n
\frac{(m(1+q))^{2i+1}}{i! (2i + 1)}
\;\leq\;
\frac{1}{2} +  \frac{1}{\sqrt{\pi}}
\sum_{i=0}^\infty
\frac{(m(1+A))^{2i+1}}{i!} \\
&   =  &
\frac{1}{2} +  \frac{m(1+A) e^{(m(1+A))^2}}{\sqrt{\pi}}
\;\leq\;
e^{2 (m(1+A))^2}.
\end{array}
\end{equation}

Let $f: [0, A] \rightarrow \R$ be a continuous non-decreasing function.
The {\em global Lipschitz constant} of $f$ is defined by
\[
L := 
\sup_{x,y \in [0, A], x < y}
\frac{f(y) - f(x)}{y - x}.
\]
If $L$ is finite, then we say that $f$ is $L$-Lipschitz.
Let $\epsilon > 0$.
For an element $x \in [0, A]$, the {\em $\epsilon$-smoothed local
Lipschitz constant of $f$ at $x$} is defined by
\[
L^\epsilon_x :=
\sup_{x,y \in f^{-1}((f(x) - \epsilon, f(x) + \epsilon)), x < y}
\frac{f(y) - f(x)}{y - x}.
\]
It is obvious that $L_x^\epsilon \leq L$. If $f$ is differentiable,
then $f'(x) \leq L^\epsilon_x$.

We now give a general proposition showing how to approximate
a continuous non-decreasing Lipschitz function by a polynomial 
of moderate degree.
\begin{proposition}
\label{prop:poly}
Let $f: [0,A] \rightarrow [0,1]$ be a continuous non-decreasing onto 
function with global Lipschitz constant $L$.
Fix $0 < \epsilon < 1$. 
Let $L_x^\epsilon$ denote the $\epsilon$-smoothed local Lipschitz 
constant of $f$ at $x$.
Let $n$ be the minimum positive odd integer satisfying
$
m A \leq \frac{\epsilon^{\frac{1}{n}} \sqrt{n}}{2},
$
where
$
m := \frac{2 L}{\epsilon} \sqrt{\ln \epsilon^{-2}}.
$
Define 
$
m_x := \frac{2 L_x^\epsilon}{\epsilon} \sqrt{\ln \epsilon^{-2}}.
$
Then there is a polynomial $p(x)$ of degree at most $2n + 1$
such that 
\[
p(x) - 2 \epsilon \leq f(x) \leq p(x) + 3 \epsilon, ~~~
-m \epsilon^2 < p'(x) < \epsilon m_x + m \epsilon^2, ~~~
\forall x \in [0,A].
\]
Moreover the sum of absolute values of the coefficients of $p(x)$,
denoted by $\alpha(p(x))$, is at most 
$e^{2 ((A+1)m)^2}$. 
\end{proposition}
\begin{proof}
Subdivide the range $[0,1]$ into $t := \lceil 1/\epsilon \rceil$ many 
closed subintervals each of length $\epsilon$ except possibly the last one
whose length $\epsilon'$ may be less than $\epsilon$. 
Denote their inverse images 
under $f$ by $I_1, I_2, \ldots, I_{t}$.
For $1 \leq i < t$, let $p_i$ be the single
point intersection of closed subintervals $I_i$ and $I_{i+1}$; define
$p_0 := 0$, $p_t := A$.
The subinterval $I_i$, $1 \leq i < t$ is of length at least 
$
\frac{\epsilon}{2 L_{p_i}^{\epsilon/2}} +
\frac{\epsilon}{2 L_{p_{i-1}}^{\epsilon/2}},
$
$I_t$ is of length at least 
$
\frac{\epsilon'}{2 L_{p_t}^{\epsilon/2}} +
\frac{\epsilon'}{2 L_{p_{t-1}}^{\epsilon/2}}.
$
Observe that $\max_i L_{p_i}^{\epsilon/2} \leq  L$.
Define the function
\[
g_1(x) := 
\epsilon \sum_{i=1}^{t - 1} s(x - p_i).
\]
Then $g_1(x) \leq f(x) \leq g_1(x) + \epsilon$ for all $x \in [0,A]$.

Define 
$
m_i := \frac{2 L_{p_i}^{\epsilon/2}}{\epsilon} \sqrt{\ln \epsilon^{-2}},
$
$1 \leq i \leq t$. Then $m \geq \max_i m_i$.
Approximate the step function $s(x - p_i)$ by the sigmoid function
$\Phi(m_i (x - p_i))$. By Equation~\ref{eq:PhiVsS},
\[
\Phi(m_i (x-p_i)) ~
\begin{array}{l l l}
= & s(x - p_i) = \frac{1}{2} & \mbox{for $x=p_i$}, \\
> & s(x - p_i) = 0 & \mbox{for $x<p_i$}, \\
< & \frac{1}{2}  & \mbox{for $x<p_i$}, \\
< & s(x - p_i)  = 1 & \mbox{for $x>p_i$}, \\
> & \frac{1}{2}  & \mbox{for $x>p_i$}, \\
> & s(x - p_i) - \epsilon^2 = 1 - \epsilon^2 
& \mbox{for $x> p_i +\frac{\epsilon}{2 L_{p_i}^{\epsilon/2}}$}, \\
< & s(x - p_i) + \epsilon^2  = \epsilon^2
& \mbox{for $x< p_i - \frac{\epsilon}{2 L_{p_i}^{\epsilon/2}}$}.
\end{array}
\]
Define the function
\[
g_2(x) := 
\epsilon \sum_{i=1}^{t - 1} \Phi(m_i (x - p_i)).
\]
It is now easy to see that
$
g_2(x) - \epsilon \leq g_1(x) \leq g_2(x) + \epsilon
$
for all $x \in [0, A]$. Thus,
\[
g_2(x) - \epsilon \leq f(x) \leq g_2(x) + 2 \epsilon ~~~
\forall x \in [0,A].
\]
Also,
\[
0 < g'_2(x) < \epsilon m_i + m \epsilon^2, ~~~
\mbox{if\ }
x \in 
[
p_i - \frac{\epsilon}{2 L_{p_i}^{\epsilon/2}}, 
p_i + \frac{\epsilon}{2 L_{p_i}^{\epsilon/2}}
]
\mbox{\ for some $i$},
\]
and 
$
0 < g'_2(x) < m \epsilon^2
$
otherwise.

We now approximate the sigmoid function $\Phi(m_i (x-p_i))$ by 
the polynomial $p_n(m_i (x - p_i))$ for 
$m_i A \leq m A < \frac{\epsilon^{\frac{1}{n}} \sqrt{n}}{2}$, $n$ odd.
From Equations~\ref{eq:PhiVsP}, \ref{eq:PhiprimeVsPprime} we get
\[
p_n(m_i (x-p_i)) ~
\begin{array}{l l l}
= & \Phi(m_i (x-p_i)) = \frac{1}{2} & \mbox{for $x=p_i$}, \\
> & \Phi(m_i (x - p_i)) & \mbox{for $0 \leq x<p_i$}, \\
< & \Phi(m_i (x -  p_i)) & \mbox{for $p_i < x \leq A$}, \\
> & \Phi(m_i (x - p_i)) - \epsilon^2 & \mbox{for $p_i \leq x \leq A$}, \\
< & \Phi(m_i (x - p_i)) + \epsilon^2 & \mbox{for $0 \leq x \leq p_i$},
\end{array}
\]
\[
p'_n(m_i (x-p_i)) ~
\begin{array}{l l l}
= & \Phi'(m_i (x-p_i)) = \frac{m_i}{\sqrt{\pi}} & \mbox{for $x=p_i$}, \\
\leq & \Phi'(m_i (x -  p_i)) & \mbox{for $0 \leq x \leq A$}, \\
> & \Phi'(m_i (x - p_i)) - \epsilon^2 & \mbox{for $0 \leq x \leq A$}.
\end{array}
\]
Define the degree $2n+1$ polynomial
\[
p(x) := 
\epsilon \sum_{i=1}^{t - 1} p_n(m_i (x - p_i)).
\]
It is now easy to see that
\[
p(x) - \epsilon^2 \leq g_2(x) \leq p(x) + \epsilon^2,
~~~
p'(x) \leq g'_2(x) \leq p'(x) + m \epsilon^2,
\]
for all $x \in [0, A]$. Thus,
\[
p(x) - 2 \epsilon \leq f(x) \leq p(x) + 3 \epsilon ~~~
\forall x \in [0,A],
\]
and
\[
-m \epsilon^2 < p'(x) < \epsilon m_i + m \epsilon^2, ~~~
\mbox{if\ }
x \in 
[
p_i - \frac{\epsilon}{2 L_{p_i}^{\epsilon/2}}, 
p_i + \frac{\epsilon}{2 L_{p_i}^{\epsilon/2}}
]
\mbox{\ for some $i$},
\]
and 
$
-m \epsilon^2 < p'(x) < m \epsilon^2
$
otherwise. Now observe that if
$
x \in 
[
p_i - \frac{\epsilon}{2 L_{p_i}^{\epsilon/2}}, 
p_i + \frac{\epsilon}{2 L_{p_i}^{\epsilon/2}}
],
$
$m_x \geq m_i$. Hence we can always say that
\[
-m \epsilon^2 < p'(x) < \epsilon m_x + m \epsilon^2 ~~~
\forall x \in [0, A].
\]

Finally by Equation~\ref{eq:alpha},
\[
\alpha(p(x)) \leq 
\epsilon \sum_{i=1}^{t-1} \alpha(p_n(m_i (x-p_i))) \leq
\epsilon \sum_{i=1}^{t-1} e^{2((A+1)m_i)^2} \leq
e^{2((A+1)m)^2}.
\]
This completes the proof of the proposition.
\end{proof}

\paragraph{Remarks:}
\ \\

\noindent
1. \ 
Any continuous non-decreasing Lipschitz
function on a closed bounded interval can be converted into a function
of the above type by translating the domain and the range and scaling
the range.

\noindent
2. \ 
A similar proposition can be proved for approximating a monotonically 
non-increasing Lipschitz function by a polynomial.

\subsection{Concentration results for Lipschitz functions}
We now state some basic definitions and facts from geometric functional
analysis that will be used in the proof of our main result.
\begin{definition}
A function $f:X \rightarrow \C $ defined over a metric 
space $X$ is said to 
be $L$-Lipschitz if
$\forall x, y \in X$ it satisfies the following inequality:
\[
\lvert f(x)-f(y) \rvert \leq L \cdot d(x,y).
\]
\end{definition}
\begin{definition}
Let $X$ be a compact metric space. An $\epsilon$-net $\mathcal{N}$ of
$X$ is a finite set of points such that for any point $x \in X$, there is
a point $x' \in \mathcal{N}$ such that $d(x,x') \leq \epsilon$.
\end{definition}
\noindent
Note that compactness guarantees that finite sized $\epsilon$-nets exist
for all $\epsilon > 0$.

We will need the following definition and fact from 
\cite{aubrun_szarek_werner_2010_main}.
\begin{definition}
A function $f:X \rightarrow \C $ defined over a normed linear 
space $X$ is said to be circled if
$f(e^{i\theta} x) = f(x)$ for all
$\theta \in \R$ and $x \in X$. 
\end{definition}
\begin{fact}
\label{fact:extension}
Let $f: X \rightarrow \R$ be a function defined on a metric space $X$.
Suppose there exists a subset $Y \subseteq X$ such that $f$ restricted
to $Y$ is $L$-Lipschitz. Then there is a function 
$\hat{f}: X \rightarrow \R$
that is $L$-Lipschitz on all of $X$ satisfying $\hat{f}(y) = f(y)$ for all
$y \in Y$. If $X$ is a normed linear space over real or complex numbers
and $f$ is circled
then the extension $\hat{f}$ is also circled. 
\end{fact}
\begin{proof}
{\bf (Sketch)}
Define 
$
\hat{f}(x) := \inf_{y \in Y} [f(y) + L d(x, y)].
$
\end{proof}

In this paper, we endow $\C^n$ with the $\ell_2$-metric and 
$\U(n)$ with the Schatten $\ell_2$-metric aka Frobenius metric. The 
following fact gives a reasonably tight upper bound on the 
size of an $\epsilon$-net of $\sphere_{\C^n}$.
\begin{fact}[\mbox{\cite[Corollary~4.2.13]{Vershynin}}]
\label{fact:net}
Let $\epsilon > 0$.  There exists an $\epsilon$-net of $\sphere_{\C^n}$ of
size less than
$
(\frac{3}{\epsilon})^{2n}.
$
\end{fact}

A fundamental result about concentration of Lipschitz functions defined
on the unit sphere or the unitary group, known as Levy's lemma, lies at 
the heart of all
proofs of Dvoretzky-type theorems via the probabilistic method.
We now state the version of Levy's lemma that will be used in this
paper.
\begin{fact}[Levy's lemma, \mbox{\cite[Corollary~4.4.28]{AGZ}}]
\label{fact:levy}
Consider the Haar probability measure on $\sphere_{\C^n}$.
Let $f: \sphere_{\C^n} \rightarrow \C$ be an $L$-Lipshitz function. 
Let $\mu := \E_x[f(x)]$ and $\lambda > 0$. Then
\[
\Pr_x(\lvert f(x) - \mu \rvert \geq \lambda) \leq 
2 \exp(-\frac{n \lambda^2}{4 L^2}).
\]
\end{fact}
An elementary proof of the above fact, without explicitly calculated 
constants, can be found in \cite[Theorem~5.1.4]{Vershynin}.

For our work, we need a measure concentration inequality like Levy's 
lemma for difference of function values on two distinct arbitrary 
points which is sensitive to the distance between those points.
Such an inequality is stated in the following fact.
\begin{fact}[\mbox{\cite[Lemma~9]{aubrun_szarek_werner_2010_main}}]
\label{fact:LevyLipschitz}
Let $f: \sphere_{\C^n} \to \C$ be a circled $L$-Lipschitz 
function.
Consider the Haar probability measure on $\U(n)$.
Then for any 
$x, y \in \sphere_{\C^n}$, $x \neq y$ and for 
any $\lambda > 0$,
\[
\Pr_U[\lvert f(Ux) - f(Uy) \rvert > \lambda] 
\leq 2 \exp(-\frac{\lambda^2 n}{8 L^2 \lVert x-y \rVert_2^2}).
\]
\end{fact} 

The derandomisation in our paper is carried out by replacing the 
Stinespring dilation unitary of a quantum channel,
 which is chosen  from the Haar measure in 
\cite{aubrun_szarek_werner_2010_main}, 
with a unitary chosen uniformly at random
from a finite cardinality approximate unitary $t$-design for a 
suitable value of $t$. The next few statements lead us to the definition
of an approximate unitary $t$-design.
\begin{definition}[\mbox{\cite[Definition~2.2]{low_2009}}]
A monomial in the entries of a matrix $U$ is of degree $(r,s)$ if 
it contains $r$ conjugated elements and $s$ unconjugated elements. The
evaluation of monomial $M$ at the entries of a matrix $U$ is denoted
by $M(U)$. We call a monomial balanced if $r = s$, 
and say that it has degree $t$ if it is of degree $(t,t)$. 
A polynomial is said to be balanced  of degree 
$t$ if it is a sum of balanced monomials of degree at most $t$.
\end{definition}
\begin{definition}[\mbox{\cite[Definition~2.3]{low_2009}}]
A probability distribution $\nu$ supported on a finite set of $d \times d$
unitary matrices is said to be an exact unitary $t$-design if 
for all balanced monomials $M$ of degree at most $t$, 
$
\E_{U \sim \nu}[M(U)] = \E_{U \sim \Haar}[M(U)].
$
\end{definition}
\begin{definition}[\mbox{\cite[Definition~2.6]{low_2009}}]
A probability distribution $\nu$ supported on a finite set of $d \times d$ 
unitary matrices is said to be 
an $\epsilon$-approximate unitary $t$-design if for all 
balanced monomials $M$ of degree at most $t$
\[
\lvert 
\E_{U \sim \nu}(M(U)) - 
\E_{U \sim \Haar}(M(U)) 
\rvert \leq 
\frac{\epsilon}{d^t}.
\]
\end{definition}

We will need the following fact.
\begin{fact}[\mbox{\cite[Lemma~3.4]{low_2009}}]
\label{fact:TvsHaar}
Let $Y: \U(n) \rightarrow \C$ be a balanced polynomial of
degree $a$ in the entries of the unitary matrix $U$ that is provided
as input. Let
$\alpha(Y)$ denote the sum of absolute values of the coefficients of $Y$.
Let $r$, $t$ be  positive integers satisfying $2ar < t$. 
Let $\nu$ be an $\epsilon$-approximate unitary $t$-design. Then 
\[
\E_{U \sim \nu}[{\lvert Y_U\rvert}^{2r}]  \leq 
\E_{U \sim \Haar}[{\lvert Y_U \rvert}^{2r}]+
\frac{\epsilon \alpha(Y)^{2r}}{n^t}.
\]
\end{fact}

\section{Sharp Dvoretzky-like theorems via stratified analysis}
\label{sec:main}
In this section, we prove our main technical results viz. sharp
Dvoretzky-like theorems for Haar measure as well as approximate 
$t$-designs using stratified analysis.
We start by proving the following two lemmas which are `baby stratified'
analogues of Fact~\ref{fact:LevyLipschitz} for Haar measure and 
approximate unitary $t$-designs.
\begin{lemma} 
\label{lem:expectationHaar} 
Let $Y: \sphere_{\C^n} \rightarrow \R$ be a circled function with
global Lipschitz constant $L_1$. Suppose that  
there exists a subset $\Omega \subseteq \sphere_{\C^n}$ such that
$Y$ restricted to $\Omega$
has a smaller Lipschitz constant $L_2$. 
Let $x, y \in \sphere_{\C^n}$.
Let $Y_x := Y(Ux)$, $Y_y := Y(Uy)$ be two correlated random variables, 
under the choice of a Haar random unitary $U$. Let $\lambda > 0$.
Then 
\[
\Pr_{U \sim \Haar}[\lvert Y_x - Y_{y}\rvert > \lambda] \leq 
2 \exp(-\frac{n \lambda^{2}}{8 L_2^2 \lVert x - y\rVert_2^2}) +
2 \Pr_{z \sim \Haar}[z \in \Omega^c].
\]
\end{lemma} 
\begin{proof} 
By Fact~\ref{fact:extension}, there is a circled function $Y'$ that
agrees with $Y$ on $\Omega$ and is $L_2$-Lipschitz on all
of $\sphere_{\C^n}$.
Define correlated random variables $Y'_x$,
$Y'_y$ in the natural manner. Then
using Fact~\ref{fact:LevyLipschitz}, we get 
\begin{eqnarray*} 
\lefteqn{\Pr_{U \sim \Haar}[\lvert Y_x - Y_{y}\rvert > \lambda]} \\
& = &
\Pr_{U \sim \Haar}[(Ux, Uy) \in \Omega \times \Omega] \cdot
\Pr_{U \sim \Haar}[
\lvert Y_x - Y_{y}\rvert > \lambda|(Ux, Uy) \in \Omega \times \Omega
] \\
&   &
{} + 
\Pr_{U \sim \Haar}[
(Ux, Uy) \not \in \Omega \times \Omega
] \cdot
\Pr_{U \sim \Haar}[
\lvert Y_x - Y_{y}\rvert > \lambda|(Ux, Uy) \not \in \Omega \times \Omega
] \\
& = &
\Pr_{U \sim \Haar}[(Ux, Uy) \in \Omega \times \Omega] \cdot
\Pr_{U \sim \Haar}[
\lvert Y'_x - Y'_{y}\rvert > \lambda|(Ux, Uy) \in \Omega \times \Omega
] \\
&   &
{} + 
\Pr_{U \sim \Haar}[
(Ux, Uy) \not \in \Omega \times \Omega
] \cdot
\Pr_{U \sim \Haar}[
\lvert Y_x - Y_{y}\rvert > \lambda|(Ux, Uy) \not \in \Omega \times \Omega
] \\
& \leq &
\Pr_{U \sim \Haar}[\lvert Y'_x - Y'_{y}\rvert > \lambda] + 
2 \Pr_{z \sim \Haar}[z \in \Omega^c] \\
& \leq &
2 \exp(-\frac{n \lambda^{2}}{8 L_2^2 \lVert x - y\rVert_2^2}) +
2 \Pr_{z \sim \Haar}[z \in \Omega^c].
\end{eqnarray*} 
This finishes the proof of the lemma.
\end{proof} 

\begin{lemma} 
\label{lem:expectationtdesign} 
Let $Y: \sphere_{\C^n} \rightarrow \R$ be a 
balanced polynomial of degree $a$ in entries of the vector $x \in \C^n$ 
that is
provided as input. Let
$\alpha(Y)$ denote the sum of absolute values of the coefficients of $Y$.
Suppose $Y$ has global Lipschitz constant $L_1$.
Suppose that  
there exists a subset $\Omega \subseteq \sphere_{\C^n}$ such that
$Y$ restricted to $\Omega$
has a smaller Lipschitz constant $L_2$. 
Let $x, y \in \sphere_{\C^n}$.
Let $Y_x := Y(Ux)$, $Y_y := Y(Uy)$ be two correlated random variables, 
under the choice of a unitary $U$ chosen uniformly at random from an 
$\epsilon$-approximate unitary $t$-design $\nu$. 
Let $r$ be a positive integer satisfying $2ar \leq t$.
Let 
$
0 < \epsilon < 
\frac{
n^{t-r} (4 r L_2^2 \lVert x - y \rVert_2^2)^r
}{\alpha(Y)^{2r}}.
$
Then 
\[
\E_{U \sim \nu}[\lvert Y_x - Y_{y}\rvert^{2r}] \leq 
3 \left(\frac{4 r L_2^2 \lVert x - y\rVert_2^2}{n}\right)^{r} +
2 \Pr_{z \sim \Haar}[z \in \Omega^c] \cdot 
(L_1^2\lVert x- y \rVert_2^2)^{r}.
\]
\end{lemma} 
\begin{proof} 
Since 
$Y_x -Y_y$ is a balanced  polynomial in the entries of the unitary
matrix $U$, from  Fact~\ref{fact:TvsHaar} we have
\[
\E_{U \sim \nu}[\lvert Y_x - Y_{y}\rvert^{2r}] 
\overset{\mathrm{a}}{\leq} 
\E_{U \sim \Haar}[\lvert Y_x - Y_{y}\rvert^{2r}] +
\frac{\epsilon \alpha(Y)^{2r}}{n^t}.
\]
By choosing 
$\epsilon$ small enough to satisfy the constraint above, we get 
$
\frac{\epsilon \alpha(Y)^{2r}}{n^t}
\overset{\mathrm{b}}{\leq} 
\left(\frac{4 r L_2^2 \lVert x - y \rVert_2^2}{n}\right)^r.
$ 
Combining (a) and (b) gives 
\[
\E_{U \sim \nu}[\lvert Y_x - Y_{y}\rvert^{2r}] 
\overset{\mathrm{c}}{\leq}
\E_{U \sim \Haar}({\lvert Y_x - Y_{y}\rvert}^{2r}) +
\left(\frac{4 r L_2^2 \lVert x - y \rVert_2^2}{n}\right)^r.
\]

Now we find 
$\E_{U \sim \Haar}[\lvert Y_x - Y_{y}\rvert^{2r}]$.
Since $Y$ is a balanced polynomial, it is circled.
By Fact~\ref{fact:extension}, there is a circled function $Y'$ such that
$Y'$ agrees with $Y$ on $\Omega$ and $Y'$ is $L_2$-Lipschitz on all
of $\sphere_{\C^n}$.
Define correlated random variables $Y'_x$,
$Y'_y$ in the natural manner. Then
\begin{eqnarray*} 
\lefteqn{\E_{U \sim \Haar}[\lvert Y_x - Y_{y}\rvert^{2r}]} \\
& = &
\Pr_{U \sim \Haar}[(Ux, Uy) \in \Omega \times \Omega] \cdot
\E_{U \sim \Haar}[
\lvert Y_x - Y_{y}\rvert^{2r}|(Ux, Uy) \in \Omega \times \Omega
] \\
&   &
{} + 
\Pr_{U \sim \Haar}[(Ux, Uy) \not \in \Omega \times \Omega] \cdot
\E_{U \sim \Haar}[
\lvert Y_x - Y_{y}\rvert^{2r}|(Ux, Uy) \not \in \Omega \times \Omega]\\ 
& = &
\Pr_{U \sim \Haar}[(Ux, Uy) \in \Omega \times \Omega] \cdot
\E_{U \sim \Haar}[
\lvert Y'_x - Y'_{y}\rvert^{2r}|(Ux, Uy) \in \Omega \times \Omega
] \\
&   &
{} + 
\Pr_{U \sim \Haar}[(Ux, Uy) \not \in \Omega \times \Omega] \cdot
\E_{U \sim \Haar}[
\lvert Y_x - Y_{y}\rvert^{2r}|(Ux, Uy) \not \in \Omega \times \Omega]\\ 
& \overset{\mathrm{d}}{\leq} &
\E_{U \sim Haar}[\lvert Y'_x - Y'_{y}\rvert^{2r}]+
2 \Pr_{z \sim \Haar}[z \in \Omega^c] \cdot 
(L_1^2\lVert x- y \rVert_2^2)^{r}.
\end{eqnarray*} 

Now we find 
$\E_{U \sim \Haar}[\lvert Y'_x - Y'_{y}\rvert^{2r}]$
using Fact~\ref{fact:LevyLipschitz} and 
Low's method \cite[Lemma~3.3]{low_2009}. 
\begin{eqnarray*} 
\lefteqn{\E_{U \sim \Haar}[\lvert Y'_x - Y'_{y}\rvert^{2r}]} \\
& = &
\int_0^\infty 
\Pr_{U \sim \Haar}[\lvert Y'_x - Y'_{y} \rvert^{2r}>\lambda] \, 
d\lambda \
\;= \;
\int_0^\infty 
\Pr_{U \sim \Haar}[\lvert Y'_x - Y'_{y} \rvert>\lambda^{1/(2r)}] \,
d\lambda \\ 
& \leq &
2 \int_0^\infty 
\exp(-\frac{n \lambda^{1/r}}{8 L_2^2 \lVert x - y\rVert_2^2}) \,
d\lambda 
\; \overset{\mathrm{e}}{\leq} \;
2 \left(\frac{4 r L_2^2 \lVert x - y\rVert_2^2}{n}\right)^{r}.
\end{eqnarray*} 

Combining inequalities (d) and (e), we have
\[
\E_{U \sim \Haar}[\lvert Y_x - Y_{y}\rvert^{2r}] \leq 
2 \left(\frac{4 r L_2^2 \lVert x - y\rVert_2^2}{n}\right)^{r} +
2 \Pr_{z \sim \Haar}[z \in \Omega^c] \cdot 
(L_1^2\lVert x- y \rVert_2^2)^{r}.
\]
Further combining with (c) gives us the desired conclusion of the lemma. 
\end{proof} 

We also need a so-called {\em chaining inequality} for probability 
similar to Dudley's 
inequality in geometric functional analysis 
\cite{aubrun_szarek_werner_2010_main,pisier_1989}. The original
Dudley's inequality bounds the expectation of the supremum, over pairs
of correlated random variables, of the difference between them
in terms of an integral, over $\eta$, of
a certain function of the size of an $\eta$-net of $\sphere_{\C^n}$.
Our chaining lemma differs from it in two important respects. 
First, instead of the expectation it bounds a tail probability of 
the supremum, over pairs
of correlated random variables, of the difference between them.
Second, it replaces the integral by a finite summation
over $\eta$-nets of $\sphere_{\C^n}$ with geometrically decreasing
$\eta$. Despite the fancy name, our chaining lemma is a simple consequence
of the union bound of probabilities. Nevertheless, it is crucial to 
proving our main result as it allows us to efficiently invoke
powerful measure concentration results in order to bound the variation of a
Lipschitz function on subspaces of $\C^n$.
\begin{lemma}[Chaining]
\label{lem:chaining}
Let $\{X_s\}_{s \in \cS}$ be a family of correlated complex valued
random variables indexed
by elements of a compact metric space $\cS$. 
Let $\lambda, L_1 > 0$. The family is said to be $L_1$-Lipschitz if
for all $s, t \in \cS$, $|X_s - X_t| \leq L_1 d(s,t)$ for
all points of the sample space.
Define $i_0$ to be the unique
integer such that the radius of $\cS$ lies in the interval
$(2^{-i_0-1}, 2^{-i_0}]$. Define 
$i_1 := \max\{i_0, \lceil \log \frac{2 L_1}{\lambda}\rceil\}$. 
Let $p : \Z \rightarrow \R_+$ be a non-decreasing 
function. Suppose the infinite series
$\sum_{i > i_0} \frac{\sqrt{|i| p(i)}}{2^i}$ is convergent with
value $C$.
Then,
\[
\Pr[\sup_{s,t \in \cS} \lvert X_s-X_t\rvert > \lambda] \leq  
\sum_{i = i_0+1}^{i_1+1}
\sum_{
(u,u^\prime) \in \cN_{i-1} \times \cN_{i}: 
d(u, u^\prime) < 2^{-i+2}
} \Pr [\lvert X_u - X_{u^\prime}\rvert > 
	\frac{\lambda \sqrt{|i| p(i)}}{4C \cdot 2^i} 
],
\]
for a sequence of $2^{-i}$-nets $\cN_i$, $i_0 \leq i \leq i_1$,
$|\cN_{i_0}| = 1$, of $\cS$. 
\end{lemma}
\begin{proof}
For every $i \in \Z$, let $\cN_i$ be a $2^{-i}$-net of $\cS$.
Let $i_0$ be such that radius of $\cS$ 
lies in $(2^{-(i_0+1)},2^{-i_0}]$. The net $\cN_{i_0}$ 
consists of a single element, say $s_0$. For every $s \in \cS$ and 
$i \in \Z$, let $\pi_i(s)$ be an element of $\cN_i$ satisfying 
$d(s,\pi_i(s)) \leq 2^{-i}$. We have the following chaining equation
for every $s \in \cS$:
\[
X_s = 
X_{s_0} +
\left(\sum_{i = i_0}^{i_i} (X_{\pi_{i+1}(s)}-X_{\pi_{i}(s)})\right) +
(X_s - X_{\pi_{i_1 + 1}(s)}).
\]
Lipschitz property of the family implies that
\begin{eqnarray*}
\sup_{s,t \in \cS} \lvert X_s-X_t\rvert 
& \leq &
2 \sum_{i = i_0}^{i_1}
\sup_{s \in \cS} \lvert X_{\pi_{i+1}(s)}-X_{\pi_{i}(s)} \rvert +
L_1 2^{-i_1} \\
& \leq &
2 \sum_{i = i_0}^{i_1}
\sup_{(u,u^\prime) \in \cN_{i} \times \cN_{i+1} : d(u,u')<2^{-i+1}} 
\lvert X_u - X_{u'} \rvert +
L_1 2^{-i_1} \\
& \leq &
2 \sum_{i = i_0+1}^{i_1+1}
\sup_{(u,u^\prime) \in \cN_{i-1} \times \cN_{i} : d(u,u')<2^{-i+2}} 
\lvert X_u - X_{u'} \rvert +
\frac{\lambda}{2}.
\end{eqnarray*}

Now if 
$\sup_{s,t \in \cS} \lvert X_s-X_t\rvert > \lambda$, there must
exist an $i$, $i_0 + 1 \leq i \leq i_1 + 1$ such that
\[
\sup_{(u,u^\prime) \in \cN_{i-1} \times \cN_{i} : d(u,u')<2^{-i+2}} 
\lvert X_u - X_{u'} \rvert > 
\frac{\lambda \sqrt{|i| p(i)}}{4C \cdot 2^i}.
\]
Applying the union bound on probability leads us to the conclusion
of the lemma.
\end{proof}

We now prove our sharp Dvoretzky-like theorem for subspaces chosen
from the Haar measure using stratified analysis.
\begin{theorem}
\label{thm:mainHaar}
Let $p : \N \rightarrow \R_+$ be a non-decreasing 
function. Suppose the infinite series
$\sum_{i > 0} \frac{\sqrt{i p(i)}}{2^i}$ is convergent with
value $C$.
Let $f: \sphere_{\C^n} \rightarrow \R$ have
global Lipschitz constant $L_1$. 
Let $L_2, c_1, c_2, c_3, \lambda > 0$. 
Define $m := \lceil \frac{c_1 n \lambda^2}{L_2^2} \rceil$.
Suppose there is an increasing sequence of subsets 
$\Omega_1 \subseteq \Omega_2 \subseteq \cdots$ of $\sphere_{\C^n}$ 
such that with probability at least $1 - c_2 e^{-c_3 m i}$,
a Haar random subspace of dimension 
$m$ lies in $\Omega_i$ and $f$ restricted to $\Omega_i$ has Lipschitz 
constant $L_2 \sqrt{p(i)}$. 
Then there exists a constant $c$ depending on $c_3$, $C$,
$0 < c < 1$, such that for
$
m^\prime := c m
$
with probability at least $1- (c_2 + 1) 2^{-m'}$, a
subspace $W$ of dimension 
$m^\prime$ chosen with respect to Haar measure
satisfies the property that
$\lvert f(w)- \mu \rvert < \lambda$ 
for all points $w \in W \cap \sphere_{\C^n}$. 
\end{theorem}
\begin{proof}
In this proof $\sphere_{\C^{n}}$ denotes the unit $\ell_2$-length sphere in
$\C^n$ together with the origin point $0$. The radius of 
$\sphere_{\C^{n}}$ is one which makes $i_0 = 0$ in 
Lemma~\ref{lem:chaining}.
Consider a canonical embedding of $\sphere_{\C^{m'}}$ into 
$\sphere_{\C^{m}}$ and further into $\sphere_{\C^n}$.  Define 
\[
B_i := \{U \in \U(n): \forall z \in \sphere_{\C^{m}}, Uz \in \Omega_i\}.
\]
For $s \in \sphere_{\C^{m'}}$, define the random variable 
$Y_s := f(Us) - \mu$,
where the randomness arises solely from the choice of $U \in \U(n)$.
Then $\Pr_{U \sim \Haar}[B_i] \geq 1 - c_2 e^{-c_3 m i}$.

Let $i_1 := \lceil \log \frac{2 L_1}{\lambda} \rceil$. 
Let $\cN_i$, $i = 0, 1, \ldots, i_1$ be a sequence of 
$2^{-i}$-nets in $\sphere_{\C^{m'}}$ of minimum cardinality, 
where $\cN_0 := \{0\}$
and $Y_0 := 0$. We can take $|\cN_i| \overset{a}{\leq} 2^{2 (i+2) m'}$ by
Fact~\ref{fact:net}. By Lemma~\ref{lem:chaining}
\[
\Pr_{U \sim \Haar}[
\sup_{s,t \in \sphere_{\C^{m'}}}
\lvert Y_s -Y_t \rvert > \lambda
] \leq 
2 \sum_{i=1}^{i_1+1} 
\sum_{
(u,u^\prime) \in \cN_{i-1} \times \cN_{i}: 
\lVert u- u^\prime \rVert_2 < 2^{-i+2}
} 
\Pr_{U \sim \Haar}[
\lvert Y_u - Y_{u^\prime}\rvert > 
\frac{\lambda \sqrt{i p(i)}}{4C \cdot 2^i}
].
\]

Applying Lemma~\ref{lem:expectationHaar} to the set $B_{i}$ gives, for
$u$, $u'$ satisfying $\lVert u- u^\prime \rVert_2 < 2^{-i+2}$,
\begin{eqnarray*}
\lefteqn{
\Pr_{U \sim \Haar}[
\lvert Y_u - Y_{u^\prime}\rvert > 
\frac{\lambda \sqrt{i p(i)}}{4C \cdot 2^i}
]
} \\
& \leq & 
2 \exp\left(-\frac{n \lambda^2 i p(i)}{2^7 C^2 2^{2i} L_2^2 p(i) 
\lVert u -u^\prime \rVert_2^2} \right) +
2 \Pr_{z \sim \Haar} [z \in \Omega_{i}^c] \\
& \leq &
2 \exp\left(-\frac{n i \lambda^2 }{2^9 C^2 L_2^2} \right) + 
2 \Pr_{z \sim \Haar} [z \in \Omega_{i}^c] \\
\\
& \leq &
2 \exp\left(-\frac{im}{2^9 C^2} \right) +
2 c_2 \exp(-c_3 m i) 
\; \leq \;
2 (c_2 + 1) \exp(-c_4 m i),
\end{eqnarray*}
for a constant $c_4$ depending only on $C$ and $c_3$.

This gives us
\begin{eqnarray*}
\lefteqn{
\Pr_{U \sim \Haar}[
\sup_{s,t \in \sphere_{\C^{m'}}}
\lvert Y_s -Y_t \rvert > \lambda
] 
} \\
& \leq & 
4 (c_2 + 1) \sum_{i=1}^{i_1+1} 
\sum_{
(u,u^\prime) \in \cN_{i-1} \times \cN_{i}:
\lVert u- u^\prime \rVert_2 < 2^{-i+2}
}  
e^{-c_4 m i} 
\;\leq\; 
4 (c_2 + 1) \sum_{i=1}^{i_1+1} 
|\cN_{i-1}| \cdot |\cN_{i}| \cdot e^{-c_4 m i} \\
& \leq &
4 (c_2 + 1) \sum_{i=1}^{i_1+1} 
2^{4 m' (i+2)} 
e^{-c_4 m i} 
\;\leq\;
(c_2 + 1) 2^{-m'},
\end{eqnarray*}
where the third inequality follows from (a) and the fourth inequality
follows from the definition $m' := c m$ for an appropriate choice
of $c$ depending only on $c_4$. In other words, $c$ depends only on
$C$ and $c_3$.

Taking $t = 0$, we see that
with probability at least $1 - (c_2 + 1) 2^{-m'}$ over
the choice of a Haar random unitary, we have
that for all $s \in \sphere_{\C^{m'}}$, $|Y_s| \leq \lambda$.
This completes the proof of the theorem.
\end{proof}

\paragraph{Remark:}
The sets $\Omega_i$ and the Lipschitz constants $L_2 \sqrt{p(i)}$
for $1 \leq i \leq \lceil \log \frac{2 L_1}{\lambda} \rceil + 1$
formalise the idea of stratified analysis mentioned intuitively in the
introduction. As $i$ increases the relevant Lipschitz constant increases.
So we need a finer net i.e. a $2^{-i}$-net for the $i$th layer 
$\Omega_i$ in order to control the variation of $f$ for subspaces 
lying inside $\Omega_i$.
With exponentially high probability, we thus get a 
Haar random subspace of dimension $m^\prime$, slightly smaller than $m$,
where $f$ is almost constant.
Note that the definition of $m$ involves only the smallest local 
Lipschitz constant $L_2$. Thus the dimension of the space $m'$ that
we obtain is larger than what would be obtained by a naive 
analysis which would be constrained by the
global Lipschitz constant $L_1$. Moreover, a naive analysis would not
give exponentially high probability, just an arbitrary constant close
one. These two properties underscore the power of our stratified
analysis. However, applying the stratified analysis to a concrete
function is not always straightforward. We need to define the
layers $\Omega_1, \Omega_2, \ldots, $ properly and show separately that
Haar random subspaces of dimension $m$ lie in $\Omega_i$ with
probability $1 - c_2 e^{-c_3 m i}$. But for several interesting functions
this can be done without much difficulty. 
This will become clearer in
Section~\ref{sec:vonNeumannentropy} where we will show how to
recover Aubrun, Szarek and Werner's result for the Haar measure directly 
from 
Theorem~\ref{thm:mainHaar}, without having to apply a Dvoretzky-style
theorem twice in a messy fashion as in the original paper
\cite{aubrun_szarek_werner_2010_main}. Moreover, we get success probability
exponentially close to one  unlike Aubrun, Szarek and Werner who could
get only a constant close to one. Furthermore, our
methods extend to approximate $t$-designs and allows us to prove
exponentially close to one probability even for that setting.

We now prove our sharp Dvoretzky-like theorem for subspaces chosen
from approximate $t$-designs using stratified analysis. 
\begin{theorem} 
\label{thm:maintdesign} 
Let $p : \N \rightarrow \R_+$ be a non-decreasing 
function. Suppose the infinite series
$\sum_{i > 0} \frac{\sqrt{i p(i)}}{2^i}$ is convergent with
value $C$.
Let $f: \sphere_{\C^n} \rightarrow \R$ be a balanced degree $`a'$ 
polynomial with global Lipschitz constant $L_1$. 
Let $0 \leq L_2 \leq 1$, $c_1, c_2, c_3, \lambda > 0$. 
Define $m := \lceil \frac{c_1 n \lambda^2}{L_2^2} \rceil$.
Suppose there is an increasing sequence of subsets 
$\Omega_1 \subseteq \Omega_2 \subseteq \cdots$ of $\sphere_{\C^n}$ 
such that with probability at least $1 - c_2 e^{-c_3 m i}$,
a Haar random subspace of dimension 
$m$ lies in $\Omega_i$ and $f$ restricted to $\Omega_i$ has Lipschitz 
constant $L_2 \sqrt{p(i)}$. 
Suppose 
\[
0 < \epsilon < 
\left(\frac{\lambda}{4 L_1}\right)^{2m} \cdot
\frac{n^{(2a-1)m} (L_2^2 p(1))^{m}}{\max\{\alpha(f)^{2m},1\}}.
\]
Then there exists a constant $c$ depending on $c_1$, $c_3$, $C$, $p(1)$, 
$0 < c < 1$ such that for
\[
m^\prime := 
c m 
\frac{\log \log \frac{C^2 L_1^2}{\lambda^2 p(1)}}
	{\lceil \log \frac{C^2 L_1^2}{\lambda^2 p(1)} \rceil},
\]
with probability at least $1- (c_2+1) 2^{-m'}$, a 
subspace $W$ of dimension 
$m^\prime$ chosen under an $\epsilon$-approximate $(2am)$-design $\nu$
satisfies the property that
$\lvert f(w)- \mu \rvert < \lambda$ 
for all points $w \in W \cap \sphere_{\C^n}$. 
\end{theorem}
\begin{proof}
In this proof $\sphere_{\C^{n}}$ denotes the unit $\ell_2$-length sphere in
$\C^n$ together with the origin point $0$. The radius of 
$\sphere_{\C^{n}}$ is one which makes $i_0 = 0$ in 
Lemma~\ref{lem:expectationtdesign}.
Consider a canonical embedding of $\sphere_{\C^{m'}}$ into 
$\sphere_{\C^{m}}$ and further into $\sphere_{\C^n}$.  Define 
\[
B_i := \{U \in \U(n): \forall z \in \sphere_{\C^{m}}, Uz \in \Omega_i\}.
\]
For $s \in \sphere_{\C^{m'}}$, define the random variable 
$Y_s := f(Us) - \mu$,
where the randomness arises solely from the choice of $U \in \U(n)$.
Then $\Pr_{U \sim \Haar}[B_i] \geq 1 - c_2 e^{-c_3 m i}$.

Let $i_1 := \lceil \log \frac{2 L_1}{\lambda} \rceil$. 
Let $\cN_i$, $i = 0, 1, \ldots, i_1$ be a sequence of 
$2^{-i}$-nets in $\sphere_{\C^{m'}}$ of minimum cardinality, 
where $\cN_0 := \{0\}$
and $Y_0 := 0$. We can take $|\cN_i| \overset{a}{\leq} 2^{2 (i+2) m'}$ by
Fact~\ref{fact:net}. By Lemma~\ref{lem:chaining}
\begin{equation}
\label{eq:chainingtdesign}
\Pr_{U \sim \nu}[
\sup_{s,t \in \sphere_{\C^{m'}}}
\lvert Y_s -Y_t \rvert > \lambda
] \leq 
2 \sum_{i=1}^{i_1+1} 
\sum_{
(u,u^\prime) \in \cN_{i-1} \times \cN_{i}: 
\lVert u- u^\prime \rVert_2 < 2^{-i+2}
} 
\Pr_{U \sim \nu}[
\lvert Y_u - Y_{u^\prime}\rvert > 
\frac{\lambda \sqrt{i p(i)}}{4C \cdot 2^i}
].
\end{equation}

Let $r$ be a positive integer such that $r (i_1+1) < m$.
Applying Lemma~\ref{lem:expectationtdesign} to the set $B_{i}$ gives, for
$u$, $u'$ satisfying $\lVert u- u^\prime \rVert_2 < 2^{-i+2}$,
\begin{eqnarray*}
\lefteqn{
\Pr_{U \sim \nu}[
\lvert Y_u - Y_{u^\prime}\rvert > 
\frac{\lambda \sqrt{i p(i)}}{4C \cdot 2^i}
]
} \\
& = & 
\Pr_{U \sim \nu}[
\lvert Y_u - Y_{u^\prime}\rvert^{2 r i} > 
\left(\frac{\lambda^2 i p(i)}{2^4 C^2 2^{2i}}\right)^{r i}
] 
\;\leq\;
\left(\frac{2^{2i+4} C^2}{\lambda^2 i p(i)} \right)^{r i}
\E_{U \sim \nu}[\lvert Y_u - Y_{u^\prime}\rvert^{2 r i}] \\
& \leq &
3 \left(\frac{2^{2i+4} C^2}{\lambda^2 i p(i)} \right)^{r i}
\left(
\left(
\frac{4 r i L_2^2 p(i) \lVert u - u^\prime \rVert_2^2}{n}
\right)^{r i} 
+ c_2 e^{-c_3 m i} \cdot (L_1^2 \lVert u- u^\prime \rVert_2^2)^{r i}
\right) \\
& \leq &
3 \left( 
\frac{2^{2i+6} C^2 r  L_2^2 \lVert u- u^\prime \rVert_2^2}
     {n \lambda^2}
\right)^{r i} + 
3 c_2 e^{-c_3 m i}
\left(
\frac{2^{2i+4} C^2 L_1^2 \lVert u- u^\prime \rVert_2^2}{\lambda^2 i p(i)}
\right)^{r i} \\
& \leq &
\underbrace{
3 \left( 
\frac{2^{10} C^2 r L_2^2}
     {n \lambda^2}
\right)^{r i}
}_{=: \mathrm{I}} + 
\underbrace{
3 c_2 e^{-c_3 m i}
\left(
\frac{2^{8} C^2 L_1^2}{\lambda^2 p(1)}
\right)^{r i}
}_{=: \mathrm{II}}.
\end{eqnarray*}

We now analyse the two terms in the above expression.
Take 
\[
r := 
\frac{c_4 n \lambda^2}{2^{10} C^2 L_2^2} \cdot 
\frac{1}{\lceil \log \frac{2^8 C^2 L_1^2}{\lambda^2 p(1)} \rceil}
\]
for a constant $c_4$, $0 < c_4 < 1$, $c_4$ depending only on
$C$, $c_1$, $c_3$, $p(1)$ chosen to be small enough so that 
$r (i_1+1) < m$ and
$\frac{c_4 n \lambda^2}{2^{10} C^2 L_2^2} \leq \frac{c_3 m}{2}$.
Substitute $r$ back in I and II to get
\[
\mathrm{I}
\leq 
3 \cdot 2^{-r i \log \log \frac{2^8 C^2 L_1^2}{\lambda^2 p(1)}},
~~~
\mathrm{II}
\leq 
3 c_2 e^{-c_3 m i} 2^{\frac{c_3 m i}{2}}
<
3 c_2 e^{-c_3 m i / 2}.
\]
We choose 
\[
m'' :=
r  \log \log \frac{2^8 C^2 L_1^2}{\lambda^2 p(1)} <
\frac{c_3 m}{2}.
\]
This gives us
\[
\mathrm{I}
\leq 
3 \cdot 2^{-m'' i},
~~~
\mathrm{II}
\leq
3 c_2 e^{-m'' i}.
\]
Thus, we have shown that
\[
\Pr_{U \sim \nu}[
\lvert Y_u - Y_{u^\prime}\rvert > 
\frac{\lambda\sqrt{p(i)}}{4 C \cdot 2^i} 
] \leq 3 (c_2 + 1) 2^{-m'' i}.
\]
Substituting above in Equation~\ref{eq:chainingtdesign}, we get
\begin{eqnarray*}
\lefteqn{
\Pr_{U \sim \nu}[
\sup_{s,t \in \sphere_{\C^{m^\prime}}}\lvert Y_s -Y_t \rvert > \lambda
]  
} \\
& \leq &
2 \sum_{i=1}^{i_1+1} 
\sum_{u,u^\prime \in \cN_{i-1} \times \cN_{i}: 
\lVert u- u^\prime \rVert < 2^{-i+2}
} 3 (c_2 + 1) 2^{-m'' i} \\
& \leq &
6 (c_2+1) \sum_{i=1}^{i_1+1} 
|\cN_{i-1}| \cdot |\cN_{i}| \cdot
2^{-m'' i} 
\;\leq\;
6 (c_2+1) \sum_{i=1}^{i_1+1} 
2^{4 m'(i+2)} 2^{-m'' i}
\;\leq\;
(c_2 + 1) 2^{-m'},
\end{eqnarray*}
if $m'$ is chosen as indicated above for a small enough constant
$c$, $0 < c < 1$, $c$ depending only on $c_4$, $c_1$, $C$ i.e.
$c$ depending only on $C$, $c_1$, $c_3$, $p(1)$.

Taking $t = 0$, we see that
with probability at least $1 - (c_2 + 1) 2^{-m^\prime}$ over
the choice of a uniformly random unitary from the approximate 
$(2am)$-design, we have
that for all $s \in \sphere_{\C^{m'}}$, $|Y_s| \leq \lambda$. 
This completes the proof of the theorem.
\end{proof}

\section{Strict subadditivity of minimum output von Neumann entropy for
approximate $t$-designs}
\label{sec:vonNeumannentropy}
We first apply Theorem~\ref{thm:mainHaar}
in order to directly recover Aubrun, Szarek and Werner's result
\cite{aubrun_szarek_werner_2010_main}
that channels with Haar random unitary Stinespring dilations
exhibit strict subadditivity of minimum output von Neumann entropy. In 
fact, we
go beyond their result in the sense that we obtain exponentially
high probability close to one as opposed to constant 
probability. After this warmup, we apply Theorem~\ref{thm:maintdesign}
in order to show that channels with approximate $n^{2/3}$-design 
unitary Stinespring dilations exhibit strict
subadditivity of minimum output von Neumann entropy with exponentially high
probability close to one.

Let $k$ be a positive integer. 
Consider the sphere $\sphere_{\C^{k^3}}$. Define the $k \times k^2$
matrix $M$ to be the rearrangment of a $k^3$-tuple from
$\sphere_{\C^{k^3}}$. Note that the $\ell_2$-norm on $\C^{k^3}$ is
the same as the Frobenius norm on $\C^{k \times k^2}$.

In Step~I, we define the 
function $f: \sphere_{\C^{k^3}} \rightarrow \R$ as 
$f(M) := \lVert M \rVert_\infty$. The function $f$ has global Lipschitz
constant $L_1 = 1$ since
\[
|f(M) - f(N)| \leq 
\lVert M - N \rVert_\infty \leq
\lVert M - N \rVert_2.
\]
For large enough $k$ the mean $\mu$ of $f$, under the Haar measure, is less
than $2 k^{-1/2}$ \cite[Corollary~7]{aubrun_szarek_werner_2010_main}.
We use the notation of Theorem~\ref{thm:mainHaar}.
Define $L_2 := 1$, $p(i) := 1$ for all $i \in \N$. Then $C < 2$.
Define the layers $\Omega_1, \Omega_2, \ldots, $ to be all of
$\sphere_{\C^{k^3}}$. Let $j$, $4 \leq j \leq k$ be a positive integer. 
Let $\lambda_j := \sqrt{\frac{j}{k}}$. Define $c_1 := 1$,
$m = k^2$, $c_2 := 0$, $c_3 := 1$. Trivially,
a Haar random subspace of dimension $m j$ lies in $\Omega_i$ with
probability at least $1 - c_2 e^{-c_3 m j i}$. Theorem~\ref{thm:mainHaar}
tells us that there is a universal constant $\hat{c}_1$ such that
for $m' := \hat{c}_1 k^2$, with probability at least
$1 - 2^{-m' j}$, a Haar random subspace $W$ of dimension $m' j$
satisfies 
\[
\lVert M \rVert_\infty < 
\frac{2}{\sqrt{k}} + \sqrt{\frac{j}{k}} <
2 \sqrt{\frac{j}{k}}
\]
for all $M \in W$.

In  Step~II, we define the 
function $f: \sphere_{\C^{k^3}} \rightarrow \R$ as 
$f(M) := \lVert M M^\dagger - \frac{\one}{k} \rVert_2$. 
The function $f$ has global Lipschitz
constant $L_1 = 2$ since
\begin{eqnarray*}
|f(M) - f(N)| 
& \leq &
\lVert M M^\dagger - N N^\dagger \rVert_2 
\; \leq \;
\lVert M M^\dagger - M N^\dagger \rVert_2 +
\lVert M N^\dagger - N N^\dagger \rVert_2 \\
& \leq &
\lVert M \rVert_\infty \lVert M^\dagger - N^\dagger \rVert_2 +
\lVert N^\dagger \rVert_\infty \lVert M - N \rVert_2 \\
&   =  &
(\lVert M \rVert_\infty + \lVert N \rVert_\infty)
\lVert M - N \rVert_2 
\;\leq\;
2 \lVert M - N \rVert_2.
\end{eqnarray*}
The mean $\mu$ of $f$, under the Haar measure, is less
than $c_0 k^{-1}$ for a universal constant $c_0$ 
\cite[Corollary~7]{aubrun_szarek_werner_2010_main}.
We use the notation of Theorem~\ref{thm:mainHaar}.
Let $j$, $c_0 < j \leq k$ be a positive integer.
Define $L_2 := 4 \sqrt{\frac{j}{k}}$, $p(i) := i+3$ for all $i \in \N$. 
Then $C \leq 4$.
Define the layers $\Omega_1, \Omega_2, \ldots, $ to be 
the subsets
\[
\Omega_i := 
\left\{
M \in \sphere_{\C^{k^3}}: 
\lVert M \rVert_\infty \leq 2 \sqrt{\frac{j(i+3)}{k}}
\right\}.
\]
It is easy to see that $f$ restricted to $\Omega_i$ has local
Lipschitz constant at most $L_2 \sqrt{p(i)}$.
Let $\lambda := \frac{j}{k}$. Define $c_1 := 16 \hat{c}_1$,
$m = \hat{c}_1 j k^2$, $c_2 := 1$, $c_3 := \ln 2$. By the previous
paragraph,
a Haar random subspace of dimension $m (i+3)$ lies in $\Omega_i$ with
probability at least 
$1 - c_2 e^{-c_3 m (i+3)} \geq 1 - c_2 e^{-c_3 m i}$. 
Theorem~\ref{thm:mainHaar}
tells us that there is a universal constant $\hat{c}_2$ such that
for $m' := \hat{c}_2 k^2$, with probability at least
$1 - 2^{-m' j}$, a Haar random subspace $W$ of dimension $m' j$
satisfies 
\[
f(M) = \lVert M M^\dagger - \frac{\one}{k}\rVert_2 <
\frac{c_0}{k} + \frac{j}{k} <
\frac{2j}{k}
\]
for all $M \in W$. Setting $j = 1$ allows us to recover Aubrun, Szarek
and Werner's technical result \cite{aubrun_szarek_werner_2010_main} with
probability exponentially close to one viz.
with probability at least
$1 - 2^{-m'}$, a Haar random subspace $W$ of dimension $m'$
satisfies 
$
\lVert M M^\dagger - \frac{\one}{k}\rVert_2 <
\frac{2}{k}
$
for all $M \in W$. We will now see how this implies the existence
of a channel with strictly subadditive minimum output von Neumann
entropy.
\begin{fact}
\label{fact:subadditivity}
Let $k$ be a positive integer.
Let $W$ be a Haar random subspace of dimension $m := \hat{c}_2 k^2$ 
chosen from the Hilbert space $\C^{k^3}$, where $\hat{c}_2$ is a
universal constant. Let $\Phi$ be the channel with output dimension
$k$ corresponding to the
subspace $W$. Then with probability at least $1 - 2^{-m}$ over the
choice of $W$,
\[
S_{\mathrm{min}}(\Phi)\geq 
\log k - \frac{4}{k}, 
~~~
S_{\mathrm{min}}(\Phi \otimes \bar{\Phi}) \leq
2 \log k - \frac{\hat{c}_2 \log k}{k} + 
O \left(\frac{1}{k} \right). 
\]
In other words,
$
S_{\mathrm{min}}(\Phi \otimes \bar{\Phi}) <
S_{\mathrm{min}}(\Phi) + S_{\mathrm{min}}(\bar{\Phi})
$
for large enough $k$.
\end{fact}
\begin{proof}
The input dimension of the channel $\Phi$ is $\dim W = m$. The 
Stinespring dilation
of the channel $\Phi$ is the $k^3 \times k^3$ unitary matrix 
that defines the subspace $W$. The subspace $W$ is obtained by
taking the first $m$ columns of a Haar random unitary matrix. 
Let $M$ be a unit $\ell_2$-norm vector in $\C^{k^3}$ rearranged as
a $k \times k^2$ matrix.
From Fact~\ref{fact:minoutputentropy}, we get
\[
S_{\mathrm{min}}(\Phi)\geq 
\log k - k \max_{M \in W} \lVert M M^\dag - \frac{\one}{k} \rVert_2^2 \geq
\log k - \frac{4}{k}.
\]
And from Fact~\ref{fact:maxeigenvalue}, with $d = k^2$, we get
\begin{eqnarray*}
S_{\mathrm{min}}(\Phi \otimes \bar{\Phi}) 
& \leq & 
2 \log k - 
\frac{m}{kd} \log k + 
O \left( \frac{m}{kd} \log \frac{d}{m} + \frac{1}{k} \right) \\
& = &
2 \log k - \frac{\hat{c}_2 \log k}{k} + 
O \left(\frac{1}{k} \right) \\
& < &
S_{\mathrm{min}}(\Phi) + S_{\mathrm{min}}(\bar{\Phi}),
\end{eqnarray*}
for large enough $k$.
\end{proof}
Thus we have shown
that for large enough $n$, Haar random $n \times n$ unitaries give rise 
to channels
exhibiting strict subadditivity of minimum output von Neumann
entropy implying that classical Holevo capacity of quantum channels can
be superadditive.

In Step~III, we define the 
function $f: \sphere_{\C^{k^3}} \rightarrow \R$ as 
$f(M) := \lVert M M^\dagger - \frac{\one}{k} \rVert_2^2$ i.e.
this $f$ is the square of the $f$ defined in Step~II above.
Now, $f$ is a balanced polynomial of degree $a = 2$ and
$1 < \alpha(f) < k^{6}$ as can be seen by considering
$f(J)$ where $J$ is the $k \times k^2$ all ones matrix.
The function $f$ has global Lipschitz
constant $L_1 = 4$ since
\begin{eqnarray*}
|f(M) - f(N)| 
& \leq &
|\lVert M M^\dagger - \frac{\one}{k} \rVert_2 -
 \lVert N N^\dagger - \frac{\one}{k} \rVert_2| \cdot
|\lVert M M^\dagger - \frac{\one}{k} \rVert_2 +
 \lVert N N^\dagger - \frac{\one}{k} \rVert_2| \\
& \leq &
(\lVert M \rVert_\infty + \lVert N \rVert_\infty)
(\lVert M M^\dagger - \frac{\one}{k} \rVert_2 +
\lVert N N^\dagger - \frac{\one}{k} \rVert_2)
\lVert M - N \rVert_2 \\
& \leq &
4 \lVert M - N \rVert_2.
\end{eqnarray*}
The mean $\mu$ of $f$ under the Haar measure is less
than $c_0^2 k^{-2}$ for the same universal constant $c_0$ 
\cite[Corollary~7]{aubrun_szarek_werner_2010_main}.
We use the notation of Theorem~\ref{thm:maintdesign}.
Define $L_2 := 16 k^{-3/2} $, $p(i) := i^3$ for all $i \in \N$. 
Then $C \leq 5$.
Define the layers $\Omega_1, \Omega_2, \ldots, $ to be 
the subsets
\[
\Omega_i := 
\left\{
M \in \sphere_{\C^{k^3}}: 
\lVert M \rVert_\infty \leq 2 \sqrt{\frac{i}{k}},
\lVert M M^\dagger - \frac{\one}{k}\rVert_2 < \frac{2i}{k}
\right\}.
\]
It is easy to see that $f$ restricted to $\Omega_i$ has local
Lipschitz constant at most $L_2 \sqrt{p(i)}$.
Let $\lambda := k^{-2}$. Define $c_1 := 2^8 \hat{c}_2$,
$m = \hat{c}_2 k^2 < \hat{c}_1 k^2$, $c_2 := 2$, 
$c_3 := \ln 2$. By the previous two
paragraphs,
a Haar random subspace of dimension $m i$ lies in $\Omega_i$ with
probability at least 
$1 - c_2 e^{-c_3 m i}$.
In particular, 
a Haar random subspace of dimension $m$ lies in $\Omega_i$ with
probability at least 
$1 - c_2 e^{-c_3 m i}$.
Let 
\[
0 \leq \epsilon < 
\left(\frac{1}{16 k^2}\right)^{2m} \frac{k^{9m} k^{-3m}}{k^{12m}} =
(4 k)^{-10 \hat{c}_2 k^2}.
\]
Theorem~\ref{thm:maintdesign}
tells us that there is a universal constant $\hat{c}_3$ such that
for 
\[
m' := \hat{c}_3 k^2 \frac{\log \log k}{\log k}, 
\]
with probability at least
$1 - 3 \cdot 2^{-m'}$, a subspace $W$ of dimension $m'$ chosen from an
$\epsilon$-approximate $(4 \hat{c}_2 k^2)$-design $\nu$ satisfies
\[
f(M) = \lVert M M^\dagger - \frac{\one}{k}\rVert_2^2 <
\frac{c_0^2}{k^2} + \frac{1}{k^2} =
\frac{c_0^2 + 1}{k^2}
\]
for all $M \in W$.
We shall now see how this result gives us a channel with strict
subadditivity of minimum output von Neumann entropy. 
\begin{theorem}
\label{thm:subadditivity}
Let $k$ be a positive integer.
Let $W$ be a subspace of dimension 
$
m' := \hat{c}_3 k^2 \frac{\log \log k}{\log k} 
$
chosen with uniform probability from a 
$k^{-8 \hat{c}_2 k^2}$-approximate unitary $(4 \hat{c}_2 k^2)$-design 
from the Hilbert space $\C^{k^3}$, where $\hat{c}_2$, $\hat{c}_3$ are
universal constants. Let $\Phi$ be the channel with output dimension
$k$ corresponding to the
subspace $W$. Then with probability at least $1 - 3 \cdot 2^{-m'}$ over the
choice of $W$,
\[
S_{\mathrm{min}}(\Phi)\geq 
\log k - \frac{c_0}{k}, 
~~~
S_{\mathrm{min}}(\Phi \otimes \bar{\Phi}) \leq
2 \log k - \frac{\hat{c}_3 \log \log k}{k} + 
O \left(\frac{(\log \log k)^2}{k \log k} + \frac{1}{k} \right), 
\]
for a universal  constant $c_0$.
In other words,
$
S_{\mathrm{min}}(\Phi \otimes \bar{\Phi}) <
S_{\mathrm{min}}(\Phi) + S_{\mathrm{min}}(\bar{\Phi})
$
for large enough $k$.
\end{theorem}
\begin{proof}
The input dimension of the channel $\Phi$ is $\dim W = m'$. The 
Stinespring dilation
of the channel $\Phi$ is the $k^3 \times k^3$ unitary matrix 
that defines the subspace $W$. The subspace $W$ is obtained by
taking the first $m'$ columns of the unitary matrix. This unitary
matrix is chosen uniformly at random from a 
$k^{-8 \hat{c}_2 k^2}$-approximate
unitary $(4 \hat{c}_2 k^2)$-design. 
Let $M$ be a unit $\ell_2$-norm vector in $\C^{k^3}$ rearranged as
a $k \times k^2$ matrix.
From Fact~\ref{fact:minoutputentropy}, we get
\[
S_{\mathrm{min}}(\Phi)\geq 
\log k - k \max_{M \in W} \lVert M M^\dag - \frac{\one}{k} \rVert_2^2 \geq
\log k - \frac{c_0^2 + 1}{k}.
\]
And from Fact~\ref{fact:maxeigenvalue}, with $d = k^2$, we get
\begin{eqnarray*}
S_{\mathrm{min}}(\Phi \otimes \bar{\Phi}) 
& \leq & 
2 \log k - 
\frac{m^\prime}{kd} \log k + 
O \left( \frac{m'}{kd} \log \frac{d}{m'} + \frac{1}{k} \right) \\
& = &
2 \log k - \frac{\hat{c}_3 \log \log k}{k} + 
O \left( \frac{(\log \log k)^2}{k \log k} + \frac{1}{k} \right) \\
& < &
S_{\mathrm{min}}(\Phi) + S_{\mathrm{min}}(\bar{\Phi}),
\end{eqnarray*}
for large enough $k$.
\end{proof}
Thus we have shown
that for large enough $n$, approximate unitary $n^{2/3}$-designs give rise 
to channels
exhibiting strict subadditivity of minimum output von Neumann
entropy, implying that classical Holevo capacity of quantum channels can
be superadditive.

\paragraph{Remark:}
Observe that the counter example we get for additivity conjecture for 
classical Holevo capacity of quantum channels, when the channel is chosen 
from an approximate unitary $t$-design has weaker parameters 
than a channel chosen from Haar random unitaries.
Nevertheless, as explained in the introduction our work 
is the first partial derandomisation of a construction of quantum
channels violating additivity of classical Holevo capacity.

\section{Strict subadditivity of minimum output R\'{e}nyi $p$-entropy for
approximate $t$-designs}
\label{sec:Renyipentropy}
In this section, we apply Proposition~\ref{prop:poly} and 
Theorem~\ref{thm:maintdesign}
in order to show that channels with approximate 
$(n^{1.7} \log n)$-design 
unitary Stinespring dilations exhibit strict
subadditivity of minimum output R\'{e}nyi $p$-entropy for
$p > 1$ with exponentially high
probability close to one.

Let $k$ be a positive integer. 
Consider the sphere $\sphere_{\C^{k^3}}$. Define the $k \times k^{2}$
matrix $M$ to be the rearrangment of a $k^3$-tuple from
$\sphere_{\C^{k^3}}$. Note that the $\ell_2$-norm on $\C^{k^3}$ is
the same as the Frobenius norm on $\C^{k \times k^{2}}$. 
Let $1 < p \leq 1.1$. 

In Step~I, we define the function $f: \sphere_{\C^{k^3}} \rightarrow \R$
as $f(M) := \lVert M \rVert_{2p}$. The function $f$ has global Lipschitz
constant $L_1 = 1$ since
\[
|f(M) - f(N)| \leq
\lVert M - N \rVert_{2p} \leq
\lVert M - N \rVert_{2}.
\]
For large enough $k$ the mean $\mu$ of $f$, under the Haar measure,
is less than $2  k^{\frac{1}{2p} - \frac{1}{2}}$ 
\cite[Section~VIII]{aubrun_szarek_werner_2010}, 
\cite[Corollary~7]{aubrun_szarek_werner_2010_main}. We use the notation
of Theorem~\ref{thm:mainHaar}.
Define $L_2 := 1$, $p(i) := 1$ for all $i \in \N$. Then $C < 2$.
Define the layers $\Omega_1, \Omega_2, \ldots, $ to be all of
$\sphere_{\C^{k^3}}$. Let $j$, $4 \leq j \leq k$ be a positive integer. 
Let $\lambda_j := j^{\frac{1}{2}} k^{\frac{1}{2p} - \frac{1}{2}}$. 
Define $c_1 := 1$,
$m = k^{2 + \frac{1}{p}}$, $c_2 := 0$, $c_3 := 1$. Trivially,
a Haar random subspace of dimension $m j$ lies in $\Omega_i$ with
probability at least $1 - c_2 e^{-c_3 m j i}$. Theorem~\ref{thm:mainHaar}
tells us that there is a universal constant $\hat{c}_1$ such that
for $m' := \hat{c}_1 k^{2 + \frac{1}{p}}$, with probability at least
$1 - 2^{-m' j}$, a Haar random subspace $W$ of dimension $m' j$
satisfies 
\[
\lVert M \rVert_{\infty} \leq 
\lVert M \rVert_{2p} < 
2  k^{\frac{1}{2p} - \frac{1}{2}} +
j^{\frac{1}{2}} k^{\frac{1}{2p} - \frac{1}{2}} <
2 j^{\frac{1}{2}} k^{\frac{1}{2p} - \frac{1}{2}}
\]
for all $M \in W$. In particular, with probability at least 
$1 - 2^{-\hat{c}_1 j k^{\frac{4}{3p} + \frac{5}{3}} (\log k)^{-1}}$, 
a Haar random 
subspace $W$ of dimension
$\hat{c}_1 j k^{\frac{4}{3p} + \frac{5}{3}} (\log k)^{-1}$ satisfies 
\[
\lVert M \rVert_{\infty} \leq 
\lVert M \rVert_{2p} < 
2 j^{\frac{1}{2}} k^{\frac{1}{2p} - \frac{1}{2}}
\]
for all $M \in W$.

Let $j$, $4 \leq j \leq k$ be a positive integer.
Define the function $f: [0, 1] \rightarrow [0, 1]$ as 
$f(x) := x^p$.
Set $\epsilon := k^{-p}$ in Proposition~\ref{prop:poly}.  Let
$n$ be the minimum positive odd integer satisfying 
$
2 p k^p \sqrt{\ln k^{2p}}  \leq
\frac{k^{-\frac{p}{n}} \sqrt{n}}{2};
$
$n < 2^7 p^3 k^{2p} \log k$.  
Proposition~\ref{prop:poly} implies
that there is a polynomial $p(x)$ of degree at most 
$2n + 1 < 2^{9} p^3 k^{2p} \log k$ such that
\begin{equation}
\label{eq:derivative}
\begin{array}{l l}
p(x) - 2 k^{-p} \leq x^p \leq p(x) + 3  k^{-p}, &
\forall x \in [0, 1], \\
|p'(x)| < 
4 p (j+1)^{p-1} \sqrt{\ln k^{2p}} 
k^{\frac{5}{3} - \frac{2}{3p} - p}, &
\forall x \in [0, j k^{\frac{2}{3p} - 1}], \\
|p'(x)| < 
4 p (5j)^{p-1} \sqrt{\ln k^{2p}} 
k^{2 - p - \frac{1}{p}}, &
\forall x \in (j k^{\frac{2}{3p} - 1}, 5 j k^{\frac{1}{p} - 1}], \\
|p'(x)| < 
4 p \sqrt{\ln k^{2p}}, &
\forall x \in (5 j k^{\frac{1}{p} - 1}, 1].
\end{array}
\end{equation}
Also, Proposition~\ref{prop:poly} guarantees that
$\alpha(p(x)) < e^{2^7 p^3 k^{2p} \log k}$.

In Step~II, we define the function $f: \sphere_{\C^{k^3}} \rightarrow \R$
as $f(M) :=  \Tr [p(M M^\dag)]$, where $p$ is the polynomial defined in 
Equation~\ref{eq:derivative}.
Now, $f$ is a balanced polynomial of degree 
$a = 2n + 1 < 2^{9} p^3 k^{2p} \log k$ and
\[
\alpha(f) = \Tr [p(J J^\dag)] = k^3 \alpha(p(x)) <  
e^{2^8 p^3 k^{2p} \log k},
\]
where $J$ is the $k \times k^2$ all ones matrix.
For a  $k \times k$ matrix $X$, define 
$\Sing(X)$ to be  the $k \times  k$ diagonal  matrix consisting of the 
singular values of $X$ arranged in decreasing order.
The function $f$ has global Lipschitz
constant $L_1 = 2^{4} p^{3/2} \sqrt{\log k}$ since
\begin{eqnarray*}
|f(M) - f(N)| 
&   =  &
|\Tr [p(\Sing(M)^2)] - \Tr [p(\Sing(N)^2)]| 
\;  = \;
|\Tr [p(\Sing(M)^2)-p(\Sing(N)^2)]| \\
& \leq &
8 p^{3/2} \sqrt{\log k} \cdot
\lVert \Sing(M)^2 - \Sing(N)^2 \rVert_1 \\
& \leq &
8 p^{3/2} \sqrt{\log k} \cdot
\lVert \Sing(M) - \Sing(N) \rVert_2 \cdot
\lVert \Sing(M) + \Sing(N) \rVert_2 \\`
& \leq &
2^{7/2} p^{3/2} \sqrt{\log k} \cdot
\lVert \Sing(M) - \Sing(N) \rVert_2 \cdot
\sqrt{\lVert M \rVert_2^2 + \lVert N \rVert_2^2} \\
& \leq &
2^{4} p^{3/2} \sqrt{\log k} \cdot
\lVert M  - N \rVert_2.
\end{eqnarray*}
Above, the first inequality follows from Equation~\ref{eq:derivative},
the second inequality is Cauchy-Schwarz and the last inequality
follows from \cite[Section~4]{Mirsky}.
By setting $j = 4$ in Step~I, we conclude that
the mean $\mu$ of $f$ under the Haar measure is less
than $2^{4p} k^{1-p}$. 
We use the notation of Theorem~\ref{thm:maintdesign}.
Let $\lambda := k^{1 - p}$. 
Define 
\[
L_2 := 
2^{4p+3} p^{3/2} 
\sqrt{\log k} \cdot 
k^{\frac{5}{3} - p - \frac{2}{3p}},
\]
$p(i) := (i+4)^{2p-1}$ for all $i \in \N$. 
Then $C \leq p^{2p}$.
Define the layers $\Omega_1, \Omega_2, \ldots, $ to be 
the subsets
\[
\Omega_i := 
\left\{
M \in \sphere_{\C^{k^3}}: 
\lVert M \rVert_{2p} \leq 
2 (i+3)^{\frac{1}{2}} k^{\frac{1}{2p}-\frac{1}{2}}
\right\}.
\]
We will now show that $f$ restricted to $\Omega_i$ has local
Lipschitz constant at most $L_2 \sqrt{p(i)}$. Note that for any
$M \in \Omega_i$, 
$
\lVert M \rVert_\infty \leq 
2 (i+3)^{\frac{1}{2}} k^{\frac{1}{2p} - \frac{1}{2}}.
$
Let $B$ denote the number of singular values of $M$ larger than
$(i+3)^{\frac{1}{2}} k^{\frac{1}{3p} - \frac{1}{2}}$. 
Let $b_1, \ldots b_k$ be the singular values of $M$ in descending order.
Then
\[
2^{2p} (i+3)^p k^{1-p} \geq
\lVert M \rVert_{2p}^{2p} \geq
\sum_{i=1}^B b_i^{2p} \geq
\left(\sum_{i=1}^B b_i^{2}\right) 
(i+3)^{p-1} k^{\frac{5}{3} - \frac{2}{3p} - p},
\]
which gives
$
\sum_{i=1}^B b_i^{2} \leq
2^{2p} (i+3) k^{\frac{2}{3p} - \frac{2}{3}}.
$
Let $C$ denote the number of singular values of $N$ larger than
$(i+3)^{\frac{1}{2}} k^{\frac{1}{3p} - \frac{1}{2}}$. 
Without loss of generality, $B \geq C$. Restricting $M$, $N$
to belong to $\Omega_i$, we get from Equation~\ref{eq:derivative} that
\begin{eqnarray*}
\lefteqn{
|f(M) -  f(N)|
} \\
& = &
|
\Tr [
p(\Sing(M)^2) - p(\Sing(N)^2)
]
|  \\
& \leq &
\sum_{i=1}^C |p(b_i^2) - p(c_i^2)| +
\sum_{i=C+1}^B |p(b_i^2) - p(c_i^2)| +
\sum_{i=B+1}^k |p(b_i^2) - p(c_i^2)| \\
& \leq &
8 p^{3/2} (5 (i+3))^{p-1} \sqrt{\log k} \cdot k^{2 - p - \frac{1}{p}} 
\sum_{i=1}^C |b_i^2 - c_i^2| \\
&     &
{} +
8 p^{3/2} (5 (i+3))^{p-1} \sqrt{\log k} \cdot k^{2 - p - \frac{1}{p}} 
\sum_{i=C+1}^B |b_i^2 - c_i^2| \\
&     &
{} +
8 p^{3/2} ((i+4))^{p-1} \sqrt{\log k} \cdot 
k^{\frac{5}{3} - p - \frac{2}{3p}}  
\sum_{i=B+1}^k |p(b_i^2) - p(c_i^2)| \\
& \leq &
8 p^{3/2} (5 (i+3))^{p-1} \sqrt{\log k} \cdot k^{2 - p - \frac{1}{p}} 
\sqrt{\sum_{i=1}^C (b_i - c_i)^2} \cdot 
\sqrt{\sum_{i=1}^C (b_i + c_i)^2} \\
&     &
{} +
8 p^{3/2} (5 (i+3))^{p-1} \sqrt{\log k} \cdot k^{2 - p - \frac{1}{p}} 
\sqrt{\sum_{i=C+1}^B (b_i - c_i)^2} \cdot 
\sqrt{\sum_{i=C+1}^B (b_i + c_i)^2} \\
&     &
{} +
8 p^{3/2} ((i+4))^{p-1} \sqrt{\log k} \cdot 
k^{\frac{5}{3} - p - \frac{2}{3p}}  
\sqrt{\sum_{i=B+1}^k (b_i - c_i)^2} \cdot 
\sqrt{\sum_{i=B+1}^k (b_i + c_i)^2} \\
& \leq &
2^{7/2} p^{3/2} 
(5 (i+3))^{p-1} \sqrt{\log k} \cdot k^{2 - p - \frac{1}{p}} 
\sqrt{\sum_{i=1}^k (b_i - c_i)^2} \cdot 
\sqrt{\sum_{i=1}^C (b_i^2 + c_i^2)} \\
&     &
{} +
2^{7/2} p^{3/2} 
(5 (i+3))^{p-1} \sqrt{\log k} \cdot k^{2 - p - \frac{1}{p}} 
\sqrt{\sum_{i=1}^k (b_i - c_i)^2} \cdot 
\sqrt{\sum_{i=C+1}^B (b_i^2 + c_i^2)} \\
&     &
{} +
2^{7/2} p^{3/2} 
((i+4))^{p-1} \sqrt{\log k} \cdot 
k^{\frac{5}{3} - p - \frac{2}{3p}}  
\sqrt{\sum_{i=1}^k (b_i - c_i)^2} \cdot 
\sqrt{\sum_{i=1}^k (b_i^2 + c_i^2)} \\
& \leq &
2^{4} p^{3/2} 2^p
5^{p-1} (i+3)^{p - \frac{1}{2}}
\sqrt{\log k} \cdot k^{2 - p - \frac{1}{p}} \cdot
k^{\frac{1}{3p} - \frac{1}{3}} \cdot
\lVert \Sing(M) - \Sing(N) \rVert_2 \\
&     &
{} +
2^{4} p^{3/2} 2^p
5^{p-1} (i+3)^{p - \frac{1}{2}}
\sqrt{\log k} \cdot k^{2 - p - \frac{1}{p}} \cdot
k^{\frac{1}{3p} - \frac{1}{3}} \cdot
\lVert \Sing(M) - \Sing(N) \rVert_2 \\
&     &
{} +
2^{4} p^{3/2} 
((i+4))^{p-1} \sqrt{\log k} \cdot 
k^{\frac{5}{3} - p - \frac{2}{3p}}  
\lVert \Sing(M) - \Sing(N) \rVert_2 \\
& \leq &
2^{6} p^{3/2} 2^p
5^{p-1} (i+4)^{p - \frac{1}{2}}
\sqrt{\log k} \cdot 
k^{\frac{5}{3} - p - \frac{2}{3p}}  \cdot
\lVert \Sing(M) - \Sing(N) \rVert_2 \\
& \leq &
2^{4p+3} p^{3/2} 
(i+4)^{p - \frac{1}{2}}
\sqrt{\log k} \cdot 
k^{\frac{5}{3} - p - \frac{2}{3p}}  \cdot
\lVert M - N \rVert_2.
\end{eqnarray*}
This completes the proof of the claim above that $f$ restricted to
$\Omega_i$ has local Lipschitz constant at most $L_2 \sqrt{p(i)}$.
Define $c_1 := 2^{8p+6} p^3 \hat{c}_1$,
$m = \hat{c}_1 k^{\frac{4}{3p}+\frac{5}{3}} (\log k)^{-1}$, $c_2 := 1$, 
$c_3 := \ln 2$. By Step~I,
a Haar random subspace of dimension $m i$ lies in $\Omega_i$ with
probability at least $1 - c_2 e^{-c_3 m i}$.
Let 
\[
0 \leq \epsilon < 
k^{\hat{c}_1 k^{2p + 3} (\log k)^{-1/2}}  <
\left(\frac{k^{1 - p}}{4 L_1}\right)^{2m} 
\frac{k^{3 (2a-1) m} 5^{(2p-1)m} L_2^{2m}}{\alpha(f)^{2m}}.
\]
Theorem~\ref{thm:maintdesign}
tells us that there is a universal constant $\hat{c}_3$ such that
for 
\[
m' := 
\hat{c}_3 k^{\frac{4}{3p} + \frac{5}{3}} \frac{\log \log k}{(\log k)^2}, 
\]
with probability at least
$1 - 2 \cdot 2^{-m'}$, a subspace $W$ of dimension $m'$ chosen from an
$\epsilon$-approximate $(2am)$-design $\nu$ 
satisfies
\[
f(M) = \Tr [p(M M^\dag)] <
2^{4p} k^{1-p} + k^{1-p} <
2^{4p+1} k^{1-p}
\]
for all $M \in W$. By Equation~\ref{eq:derivative}, this implies that
\[
\Tr [(M M^\dag)^p] < 
\Tr [p(M M^\dag)] + 3 k^{1-p} <
2^{4p+3} k^{1-p}
\]
for all $M \in W$. In other words,
$
\lVert M \rVert_{2p}^2 < 2^7 k^{\frac{1}{p} - 1}
$
for all $M \in W$.
We shall now see how this result gives us a channel with strict
supermultiplicativity of the $\lVert \cdot \rVert_{1 \rightarrow p}$-norm
or equivalently, strict
subadditivity of minimum output R\'{e}nyi $p$-entropy for any $p > 1$.
\begin{theorem}
\label{thm:supermultiplicativity}
Let $k$ be a positive integer. Let $1 < p \leq 1.1$.
Let $W$ be a subspace of dimension 
$
m' := 
\hat{c}_3 k^{\frac{4}{3p} + \frac{5}{3}} \frac{\log \log k}{(\log k)^2} 
$
chosen with uniform probability from a 
$k^{\hat{c}_1 k^{5} (\log k)^{-1/2}}$-approximate unitary 
$(2^{11} \hat{c}_1 k^{5.1})$-design 
from the Hilbert space $\C^{k^3}$, where $\hat{c}_1$, $\hat{c}_3$ are
universal constants. Let $\Phi$ be the channel with output dimension
$k$ corresponding to the
subspace $W$. Then with probability at least $1 - 2 \cdot 2^{-m'}$ over the
choice of $W$,
\[
\lVert \Phi \rVert_{1 \rightarrow p} \leq
2^{7} k^{\frac{1}{p} - 1}, 
~~~
\lVert \Phi \otimes \bar{\Phi} \rVert_{1 \rightarrow p} \geq
\hat{c}_3 k^{\frac{4}{3p} - \frac{4}{3}}.
\]
In other words,
$
\lVert \Phi \otimes \bar{\Phi} \rVert_{1 \rightarrow p} >
\lVert \Phi \rVert_{1 \rightarrow p} \cdot
\lVert \bar{\Phi} \rVert_{1 \rightarrow p} \cdot
$
for large enough $k$. For $p > 1.1$, the  channel $\Phi$ obtained
for $p = 1.1$ suffices to show supermultiplicativity.
\end{theorem}
\begin{proof}
The input dimension of the channel $\Phi$ is $\dim W = m'$. The 
Stinespring dilation
of the channel $\Phi$ is the $k^3 \times k^3$ unitary matrix 
that defines the subspace $W$. The subspace $W$ is obtained by
taking the first $m'$ columns of the unitary matrix. This unitary
matrix is chosen uniformly at random from a
$k^{\hat{c}_1 k^{5} (\log k)^{-1/2}}$-approximate unitary 
$(2^{11} \hat{c}_1 k^{5.1})$-design. Note that
$
2 a m < 2^{11} \hat{c}_1 k^{5.1},
$
$
\epsilon > k^{\hat{c}_1 k^{5} (\log k)^{-1/2}},
$
where $a$, $m$ and $\epsilon$ are defined in Step~III above.
Let $M$ be a unit $\ell_2$-norm vector in $\C^{k^3}$ rearranged as
a $k \times k^2$ matrix.
From Equation~\ref{eq:redefpnorm}, we get
\[
\lVert \Phi \rVert_{1 \rightarrow p} =
\max_{M \in W: \lVert M \rVert_2 = 1}
\lVert M \rVert_{2p}^2 \leq
2^{7} k^{\frac{1}{p} - 1}. 
\]
From Fact~\ref{fact:maxeigenvalue}, 
\[
\lVert \Phi \otimes \bar{\Phi} \rVert_{1 \rightarrow p} \geq
\lVert \Phi \otimes \bar{\Phi} \rVert_{1 \rightarrow \infty} \geq
\frac{m'}{k^3} =
\hat{c}_3 k^{\frac{4}{3p} - \frac{4}{3}} 
\frac{\log \log k}{(\log k)^2} >
(\lVert \Phi \rVert_{1 \rightarrow p})^2
\]
for large enough $k$.
This shows the supermultiplicativity of the 
$\lVert \cdot \rVert_{1 \rightarrow p}$-norm for $1 < p \leq 1.1$. For
$p > 1.1$, we use the fact that 
$ 
\lVert \cdot \rVert_{1 \rightarrow \infty} \leq
\lVert \cdot \rVert_{1 \rightarrow p} \leq
\lVert \cdot \rVert_{1 \rightarrow 1.1}
$
to  conclude the supermultiplicativity of 
$\lVert \cdot \rVert_{1 \rightarrow p}$.
\end{proof}
Thus by setting $p = 1.1$, we see
that for large enough $n$ approximate unitary 
$(n^{1.7} \log n)$-designs give 
rise to channels
exhibiting strict subadditivity of minimum output R\'{e}nyi $p$-entropy
for any $p > 1$. Combined with the result  of the previous section,
we can furthermore state 
that for large enough $n$ approximate unitary 
$(n^{1.7} \log n)$-designs give 
rise to channels
exhibiting strict subadditivity of minimum output R\'{e}nyi $p$-entropy
for any $p \geq 1$. 

\paragraph{Remarks:}

\noindent
1.\ 
In \cite{aubrun_szarek_werner_2010}, for channels obtained from
Haar random subspaces the lower bound on 
$
\lVert \Phi \otimes \bar{\Phi} \rVert_{1 \rightarrow p}
$
was of the order of
$
k^{\frac{1}{p} - 1},
$
whereas in our work it is 
of the order of
$
k^{\frac{4}{3p} - \frac{4}{3}},
$
for channels obtained from approximate $t$-designs.
Hence the counter example we get for additivity of minimum output
R\'{e}nyi $p$-entropy of quantum channels, when the channel is 
chosen 
from an approximate unitary $t$-design has weaker parameters 
than the Haar random channels of
\cite{aubrun_szarek_werner_2010}. Nevertheless, our work 
is the first partial derandomisation of a construction of quantum
channels violating additivity of minimum output
R\'{e}nyi $p$-entropy, since it is possible to uniformly
sample a unitary from an exact
$(n^{1.7} \log n)$-design 
using of the order of
$
n^{1.7} (\log n)^2
$
random bits versus $\Omega(n^2)$ random bits required to choose 
a Haar random unitary to constant precision.

\smallskip

\noindent
2.\  
It is possible to do the above counterexample on a sphere in $\C^{k^2}$. 
However
in that case the number of random bits required to choose a
unitary from an exact design is larger than $k^4 \log k$, which is what a
Haar random unitary would require! 

\smallskip

\noindent
3.\ 
It is worthwhile to note that, for channels $\Phi \text{ and } \bar{\Phi}$, with output dimension $k$, gap between the upper bound to $S_{min}(\Phi \otimes \bar{\Phi})$ and the lower bound to $S_{min}(\Phi)+S_{min}(\bar{\Phi})$ is of the order $O(\frac{\log k}{k} - \frac{1}{k})$ when channel $\Phi$ is chosen from Haar measure, 
whereas this gap is $O(\frac{\log \log k}{k} - \frac{1}{k})$ when channel is chosen uniformly from an approximate unitary $t$-design. So as the output dimension of the channels increase, the counter example for $p=1$, becomes weak as the above mentioned gap of $O(\frac{\log k}{k} - \frac{1}{k}) \text{ or } O(\frac{\log \log k}{k} - \frac{1}{k})$ decreases with increasing $k$. Also, in order to keep the constants showing up in the bound under control, one has to optimize over $k$, the output dimension of the channels. However, our counter example and for instance the existing counter example of Aubrun, Szarek and Werner \cite{aubrun_szarek_werner_2010_main} are still good enough even for arbitrarily large output dimensions. To see this, let channels $\Phi \text{ and } \bar{\Phi}$ violate the additivity of minimum output von Neumann entropy for an optimised output dimension $k$. This fixes the gap between the upper bound to $S_{min}(\Phi \otimes \bar{\Phi})$ and the lower bound to $S_{min}(\Phi)+S_{min}(\bar{\Phi})$, according to whether the channels are chosen in a Haar random fashion or uniformly from an approximate $t$-design. In order to show that the counter example does not deteriorate with random channels of higher dimension, we define new channel as $\Phi' := \Phi \otimes \Phi_{Deopolarising}$ where $\Phi_{Depolarising}$ is the completely depolarising channel given by $\Phi_{Depolarising}^{A \to B} (\rho^{A}):= \pi^B$ and $(\Phi_{Depolarising}^{A \to B} \otimes \I^C) (\rho^{AC}):= \pi^B \otimes \rho^C$, for all density operators $\rho^A, \rho^{AC}$. Hence the output dimension of the channel $\Phi'$ is $k.d_{Depolarising}$, where $d_{Depolarising}$ is the output dimension of $\Phi_{Depolarising}$. This trick works as $S_{min}(\phi')=\log d_{Depolarising}+S_{min}(\Phi)$. This ensures that the gap between the upper bound to $S_{min}(\Phi' \otimes \bar{\Phi'})$ and the lower bound to $S_{min}(\Phi')+S_{min}(\bar{\Phi'})$ is still the same as that for channels $\Phi \text{ and } \bar{\Phi}$.   
Due to the small magnitude of this gap coming up with an explicit formula or an efficient construction for additivity of minimum output von Neumann entropy is a much harder task than that for minimum output R\'enyi $p$ entropy for $p \neq 1$, as evident from \cite{Nechita}, which shows that in order to achieve ``almost'' one bit violation to additivity the output dimension of the channel should be at least $183$.

\smallskip

\noindent
4.\
On the other hand, the gap between the upper bound for $S_p^{min}(\Phi \otimes \bar{\Phi})$ and the lower bound for $S_p^{min}(\Phi)+S_p^{min}(\bar{\Phi})$ is $O(\log k)$ for $p>1$ , when the channels $\Phi, \bar{\Phi}$ are chosen according to Haar measure and $O(2/3 \log k)$ for channels being chosen from an approximate $t$-design. Here also $k$ is the output dimension of the channels. There are two things to note:
\begin{itemize}
\item This gap is much larger than that for the case when $p=1$ or for counter example to minimum output von Neumann entropy.
\item This gap does not decrease with increasing $k$
\end{itemize}  
These points suggests that it is more feasible to come up with an explicit formula for channels violating additivity of minimum output R\'enyi $p$ entropy, for $p>1$. Moreover, some known and efficient examples are that of the so-called Werner-Holevo channel \cite{WernerHolevo} for $p > 4.79$ and channels arising from an antisymmetric subspace in \cite{Grudka_2010} for $p>2$. An efficient counter for $p$ close to and equal to $0$, is also given in \cite{Cubitt_additivity}, for a channel with input dimension $4$ and output dimension $3$. We note that we do not have a unique pair of channels $\Phi, \Psi$ that serves as counter example to additivity for R\'enyi $p$-entropy for all $p \geq 0$.

Standard continuity arguments can be used to show that there exists $1> \epsilon > 0$ that violates the additivity conjecture for R\'enyi $p$-entropy for $1<p<1+\epsilon$, for the channel being described by a random unitary operator from an approximate $k^2$-design. However this $\epsilon$ is not explicitly known and depends on $k$. Finding this $\epsilon$ shall involve, restricting the value of $k$ to be small so that the gap between the bounds is large enough. But for keeping constants in the several bounds tractable demands $k$ to be large enough and hence there is no way to derandomise Aubrun, Szarek and Werner counter example for $p>1$ as in \cite{aubrun_szarek_werner_2010}. 
We thus take a first step towards this goal of derandomisation.  Our construction, via polynomial approximation of a monotonous function, in theorem~\ref{thm:supermultiplicativity} gives an explicit value of $\epsilon=0.1$ and guarantees that the gap between the upper bound for $S_p^{min}(\Phi \otimes \bar{\Phi})$ and the lower bound for $S_p^{min}(\Phi)+S_p^{min}(\bar{\Phi})$ is $O(2/3 \log k)$, which is larger than that for von Neumann entropy. This comes at the cost of a larger value of $t=k^{5.1}$ for an approximate unitary $t$-design.
This restriction on the value of $\epsilon$ and the higher value of $t$, stems from approximating the square of Schatten $2p$-norm  with a moderate degree polynomial, in order to replace a Haar random unitary via a $t$-design unitary operator. Hence, even for $p>1$ we have only partial derandomisation.

\section{Conclusion}
\label{sec:conclusion}
In this paper we have shown that a unitary chosen from an approximate 
unitary $n^{2/3}$-design leads
to a quantum channel with superadditive classical Holevo capacity. In the 
process of coming up
with such a channel we developed two new technical tools viz. 
stratified analysis 
of a sphere in $\C^n$ for Haar measure and unitary designs
(Theorems~\ref{thm:mainHaar}, \ref{thm:maintdesign}),
and approximation of any continuous monotonic function by
a polynomial of moderate degree (Proposition~\ref{prop:poly}). 
The stratified analysis for the Haar measure was used to
recover in a simple fashion Aubrun, Szarek and Werner's counterexample
\cite{aubrun_szarek_werner_2010_main} for
additivity of minimum output Von Neumann entropy. The stratified
analysis for unitary designs was used to prove
counterexamples for additivity of minimum output von Neumann entropy and
R\'{e}nyi $p$-entropy for $p > 1$, when the unitary Stinespring dilation
of the channel is chosen from approximate unitary $t$-design for
suitable values of $t$. Choosing a unitary from these $t$-designs requires
less random bits than choosing from  the Haar measure. 
However the value of $t$ required is much larger than what is known to
be efficiently implementable by quantum circuits.
We believe our work results in a better understanding of the interplay
between geometric functional analysis and additivity questions in
quantum information theory, and our technical tools will find applications
to other problems in quantum information theory.

Our  work represents a step in the  quest for an efficient explicit 
channel violating additivity of minimum output von Neumann entropy. This
is the major open problem in the area. Another problem left open is
whether there is a single channel that violates additivity of
minimum output R\'{e}nyi $p$-entropy for all $p \geq 1$.

\bibliography{additivity}

\newcommand{\etalchar}[1]{$^{#1}$}
\begin{thebibliography}{KMNR05}

\bibitem[AGZ09]{AGZ}
{Anderson, G.}, {Guionnet, A.}, and {Zeitouni, O.}
\newblock {\em An introduction to random matrices}.
\newblock Cambridge University Press, 2009.

\bibitem[AHW00]{AmosovEtAl}
G.~G. {Amosov}, A.~S. {Holevo}, and R.~F. {Werner}.
\newblock {On some additivity problems in quantum information theory}.
\newblock {\em arXiv e-prints}, pages math--ph/0003002, March 2000.

\bibitem[ASW10a]{aubrun_szarek_werner_2010_main}
Guillaume Aubrun, Stanisław Szarek, and Elisabeth Werner.
\newblock Hastings’s additivity counterexample via dvoretzky’s theorem.
\newblock {\em Communications in Mathematical Physics}, 305(1):85–97, 2010.

\bibitem[ASW10b]{aubrun_szarek_werner_2010}
Guillaume Aubrun, Stanisław Szarek, and Elisabeth Werner.
\newblock Nonadditivity of {R\'{e}nyi} entropy and {Dvoretzky’s} theorem.
\newblock {\em Journal of Mathematical Physics}, 51(2):022102, 2010.

\bibitem[BCN13]{Nechita}
Serban Belinschi, Benoit Collins, and Ion Nechita.
\newblock Almost one bit violation for the additivity of the minimum output
  entropy.
\newblock {\em Communications in Mathematical Physics}, 341, 05 2013.

\bibitem[BDSW96]{Bennett_EtAl}
Charles~H. Bennett, David~P. DiVincenzo, John~A. Smolin, and William~K.
  Wootters.
\newblock Mixed-state entanglement and quantum error correction.
\newblock {\em Phys. Rev. A}, 54:3824--3851, Nov 1996.

\bibitem[BHH16]{brandao2012local}
{Brand\~{a}o, F.}, {Harrow, A.}, and {Horodecki, M.}
\newblock Local random quantum circuits are approximate polynomial-designs.
\newblock {\em Communications in Mathematical Physics}, 346(2):397--434, 2016.

\bibitem[CHL{\etalchar{+}}07]{Cubitt_additivity}
Toby Cubitt, Aram Harrow, Debbie Leung, Ashley Montanaro, and Andreas Winter.
\newblock Counterexamples to additivity of minimum output p-r\'enyi entropy for
  $p$ close to $0$.
\newblock {\em Communications in Mathematical Physics}, 284, 12 2007.

\bibitem[Dvo61]{Dvoretzky}
Aryeh Dvoretzky.
\newblock Some results on convex bodies and {B}anach spaces.
\newblock In {\em Proc. {I}nternat. {S}ympos. {L}inear {S}paces ({J}erusalem,
  1960)}, pages 123--160. Jerusalem Academic Press, Jerusalem; Pergamon,
  Oxford, 1961.

\bibitem[FH02]{FujiwaraEtAl}
Akio Fujiwara and Takashi Hashizumé.
\newblock Additivity of the capacity of depolarizing channels.
\newblock {\em Physics Letters A}, 299(5):469 -- 475, 2002.

\bibitem[GHP10]{Grudka_2010}
Andrzej Grudka, Micha{\l} Horodecki, and {\L}ukasz Pankowski.
\newblock Constructive counterexamples to the additivity of the minimum output
  r{\'{e}}nyi entropy of quantum channels for all $p>2$.
\newblock {\em Journal of Physics A: Mathematical and Theoretical},
  43(42):425304, oct 2010.

\bibitem[Gor85]{Gordon}
Yehoram Gordon.
\newblock Some inequalities for {Gaussian} processes and applications.
\newblock {\em Israel Journal of Mathematics}, 50(4):265--289, Dec 1985.

\bibitem[Has09]{hastings_2009}
M.~B. Hastings.
\newblock Superadditivity of communication capacity using entangled inputs.
\newblock {\em Nature Physics}, 5(4):255–257, 2009.

\bibitem[HW08]{hayden_winter_2008}
Patrick Hayden and Andreas Winter.
\newblock Counterexamples to the maximal p-norm multiplicativity conjecture for
  all $p>1$.
\newblock {\em Communications in Mathematical Physics}, 284(1):263–280, Oct
  2008.

\bibitem[Kin02]{KingUnital}
Christopher King.
\newblock Additivity for unital qubit channels.
\newblock {\em Journal of Mathematical Physics}, 43(10):4641--4653, 2002.

\bibitem[Kin03]{KingDepolarising}
C.~King.
\newblock The capacity of the quantum depolarizing channel.
\newblock {\em IEEE Transactions on Information Theory}, 49(1):221--229, Jan
  2003.

\bibitem[KMNR05]{KingEtAl}
Christopher {King}, Keiji {Matsumoto}, Michael {Nathanson}, and Mary~Beth
  {Ruskai}.
\newblock {Properties of Conjugate Channels with Applications to Additivity and
  Multiplicativity}.
\newblock {\em arXiv e-prints}, pages quant--ph/0509126, September 2005.

\bibitem[Kup06]{Kuperberg}
G.~Kuperberg.
\newblock Numerical cubature from archimedes' hat-box theorem.
\newblock {\em SIAM Journal on Numerical Analysis}, 44(3):908--935, 2006.

\bibitem[Low09]{low_2009}
R.~A. Low.
\newblock Large deviation bounds for k-designs.
\newblock {\em Proceedings of the Royal Society A: Mathematical, Physical and
  Engineering Sciences}, 465(2111):3289–3308, May 2009.

\bibitem[Mil92]{Milman}
V.~Milman.
\newblock A new proof of the theorem of {A. Dvoretzky} on sections of convex
  bodies.
\newblock {\em Func. Anal. Appl.}, 5:28--37, 1992.
\newblock English translation.

\bibitem[Mir60]{Mirsky}
L.~Mirsky.
\newblock Symmetric gauge functions and unitarily invariant norms.
\newblock {\em Quarterly Journal of Mathematics Oxford}, 11(1):50--59, 1960.

\bibitem[ON00]{Osawa}
Susumu Osawa and Hiroshi Nagaoka.
\newblock Numerical experiments on the capacity of quantum channel with
  entangled input states.
\newblock {\em IEICE Transactions on Fundamentals of Electronics Communications
  and Computer Sciences}, E84A, 08 2000.

\bibitem[Pis89]{pisier_1989}
Gilles Pisier.
\newblock {\em The volume of convex bodies and Banach space geometry}.
\newblock Cambridge University Press, 1989.

\bibitem[Pom03]{Pomeransky}
A.~A. Pomeransky.
\newblock Strong superadditivity of the entanglement of formation follows from
  its additivity.
\newblock {\em Phys. Rev. A}, 68:032317, Sep 2003.

\bibitem[Sch88]{Schechtman}
G.~Schechtman.
\newblock A remark concerning the dependence on $\epsilon$ in {Dvoretzky's}
  theorem.
\newblock In {\em Geometric Aspects of Functional Analysis Israel Seminar
  (GAFA) 1987–88}, pages 274--277. Springer, Berlin, Heidelberg, 1988.

\bibitem[Sen18]{sen:zigzag}
P.~Sen.
\newblock Efficient quantum tensor product expanders and unitary $t$-designs
  via the zigzag product.
\newblock arXiv preprint arXiv:1808.10521, 2018.

\bibitem[Sho02]{ShorEntang}
Peter~W. Shor.
\newblock Additivity of the classical capacity of entanglement-breaking quantum
  channels.
\newblock {\em Journal of Mathematical Physics}, 43(9):4334--4340, 2002.

\bibitem[Sho04]{Shor}
Peter~W. Shor.
\newblock Equivalence of additivity questions in quantum information theory.
\newblock {\em Communications in Mathematical Physics}, 246(3):473--473, Apr
  2004.

\bibitem[Ver18]{Vershynin}
R.~Vershynin.
\newblock {\em High dimensional probability}.
\newblock Cambridge University Press, 2018.

\bibitem[VSW50]{vajda_shannon_weaver_1950}
S.~Vajda, Claude~E. Shannon, and Warren Weaver.
\newblock The mathematical theory of communication.
\newblock {\em The Mathematical Gazette}, 34(310):312, 1950.

\bibitem[WH02]{WernerHolevo}
Reinhard Werner and Alexander Holevo.
\newblock Counterexample to an additivity conjecture for output purity of
  quantum channels.
\newblock {\em Journal of Mathematical Physics}, 43, 04 2002.

\end{thebibliography}
\end{document}